\documentclass[11pt]{article}
\usepackage{physics}
\usepackage{comment}
\usepackage{amssymb}
\usepackage{mathrsfs}
\usepackage{enumitem}
\usepackage{graphics}
\usepackage{amsmath}
\usepackage{makecell}
\usepackage{amsmath,amsthm}
\usepackage{fullpage}
\usepackage{hyperref}
\usepackage[numbers]{natbib}
\usepackage{setspace}
\usepackage[demo]{graphicx}
\usepackage{soul}
\usepackage{graphicx}
\usepackage{xcolor}
\usepackage{url}

\usepackage{booktabs} 
\usepackage[linesnumbered,ruled,noend]{algorithm2e}

\SetAlFnt{\small}
\SetAlCapFnt{\small}
\SetAlCapNameFnt{\small}
\SetAlCapHSkip{0pt}
\IncMargin{-\parindent}

\usepackage{xcolor}

\newcounter{LP}

\newtheorem{theorem}{Theorem}[section]
\newtheorem*{theorem*}{Theorem}
\newtheorem{lemma}[theorem]{Lemma}
\newtheorem{definition}[theorem]{Definition}

\newtheorem{corollary}[theorem]{Corollary}
\newtheorem{proposition}[theorem]{Proposition}
\newtheorem{observation}[theorem]{Observation}
\newtheorem{assumption}[theorem]{Assumption}

\newtheorem{example}[theorem]{Example}
\theoremstyle{remark}
\newtheorem*{remark}{Remark}
\newtheorem*{notation}{Notation}

\usepackage{caption}
\usepackage{xcolor}

\setcitestyle{authoryear}

\setcitestyle{acmnumeric}

\newcommand{\Dis}{{\mathrm{Dis}}}
\newcommand{\Con}{{\mathrm{Con}}}

\begin{document}
\title{Contracts with Private Cost per Unit-of-Effort\thanks{An extended abstract appeared in the \emph{Proceedings of the 22nd ACM Conference on Economics and Computation}. This research was supported by the Israel Science Foundation (grant no.~336/18) and by the Taub Family Foundation. }}

\author{ Tal Alon\thanks{Technion -- Israel Institute of Technology. Email: alontal@campus.technion.ac.il.}
\and Paul D\"utting\thanks{Google Research. Email: duetting@google.com.}  
\and Inbal Talgam-Cohen\thanks{Technion -- Israel Institute of Technology. Email: italgam@cs.technion.ac.il.}
}


\date{}


\maketitle

\begin{abstract}
Economic theory distinguishes between principal-agent settings in which the agent has a private type and settings in which the agent takes a hidden action. Many practical problems, however, involve aspects of both. For example, brand X may seek to hire an influencer Y to create sponsored content to be posted on social media platform Z. This problem has a hidden action component (the brand may not be able or willing to observe the amount of effort exerted by the influencer), but also a private type component (influencers may have different costs per unit-of-effort).

This ``effort'' and ``cost per unit-of-effort'' perspective naturally leads to a principal-agent problem with hidden action and single-dimensional private type, which generalizes both the classic principal-agent hidden action model of contract theory \`a la \cite{GrossmanHart83} and the (procurement version) of single-dimensional mechanism design \`a la \cite{Myerson81}. 
A natural goal in this model is to design an incentive-compatible contract, which consist of an allocation rule that maps types to actions, and a payment rule that maps types to payments for the stochastic outcomes of the chosen action.

Our main contribution is a linear programming (LP) duality based characterization of implementable allocation rules for this model, which applies to both discrete and continuous types. This characterization shares important features of Myerson's celebrated characterization result, but also departs from it in significant ways. We present several applications, including a polynomial-time algorithm for finding the optimal contract with a constant number of actions. This is in sharp contrast to recent work on hidden action problems with multi-dimensional private information, which has shown that the problem of computing an optimal contract for constant numbers of actions is APX-hard.


\end{abstract}








\section{Introduction}

Mechanism design (MD) \cite[e.g.,][]{Myerson81} is a corner stone of economic theory, with major successes in both theory and practice.
MD studies mechanisms for resource allocation among agents with private information, known as \emph{hidden types}, with the goal of maximizing social welfare or revenue. 
The past two decades have seen a surge of interest in the study of mechanism design through an algorithmic and computational lens. 
Major success stories range from computational advertising and spectrum auctions to applications in internet routing and loadbalancing. 
An equally important and central branch in economics is contract theory (CT) \cite[e.g.,][]{GrossmanHart83}. While the natural focus of mechanism design is on the allocation of goods, contract theory has a natural focus on the \emph{allocation of effort}. Contracts are a main tool for effort allocation since they use payments (monetary or other) to determine which actions agents will take. 
While the computer science community has largely ignored contract theory, driven by increased practical demand, there is a growing momentum and a recent set of papers have started to explore applications of contract theory from a computational viewpoint \cite[e.g.,][]{babaioff2006combinatorial,DuttingRT19,DuttingRT20,guruganesh2020contracts}.

The increased practical demand for a computational and algorithmic approach to contracts is caused by an accelerating movement of contract-based markets from the analog/pen-and-paper world to the digital/electronic world. This includes online markets for crowdsourcing, sponsored content creation, affiliate marketing, freelancing and more. 
The economic value of these markets is substantial.\footnote{For example, according to Statista, influencer marketing on Instagram was worth 5.67 billion U.S.~dollars in 2018. See \url{https://www.statista.com/statistics/950920/global-instagram-influencer-marketing-spending/}.}
In such applications, platforms serve as a bridge between the two sides of the market (e.g., advertising brands and content creators), and are thus well-situated to play the role of market makers, applying their position and data to design better contracts. 

\paragraph{\bf The need for CT $\times$ MD.}

Our work is motivated by the fact that in many of the applications that motivate the surge of interest in contracts, we actually see features of both --- contract theory and mechanism design; and while there is some work on this in economics (which we survey in Section~\ref{sec:related-work}) --- with the exception of \cite{guruganesh2020contracts} --- we are not aware of any prior work on the combination of the two from the computer science perspective.

The starting point of our work is that many of the applications at the intersection of CT and MD naturally involve an agent who can exert different levels of effort (his actions), and as in classic contract theory this leads to a stochastic outcome/reward to the principal, but the cost per unit-of-effort may differ between agents and is naturally modelled as private information. For example, a brand (the principal) may approach an influencer on a social media platform (the agent) to create branded or sponsored content on their behalf; and the opportunity cost of different influencers for different amounts of effort and outcome levels may differ.

The point is that this naturally leads to a model that combines the classic-principal agent problem with hidden action with features of \emph{single-dimensional} mechanism design that we propose and study in this paper.

\paragraph{\bf Single vs.~multi-parameter.}



To concretely discuss agent types and explain where our contribution diverges from \cite{guruganesh2020contracts}, we briefly introduce the classic model for a principal-agent contractual relation \cite[e.g.,][]{GrossmanHart83,carroll2015robustness}. A basic principal-agent setting is described by $n$ (hidden) actions the agent can choose among, $m$ (observable) possible outcomes the actions can lead to, and an $n\times m$ matrix~$F$ whose $i$th row maps action $i$ to a distribution over the outcomes. In addition, there is a cost vector with $n$ costs, one per action, specifying how much costly effort the agent must invest to take that action. Finally there is a reward vector with $m$ rewards, one per outcome, specifying how much the principal gains when the agent achieves that outcome. Given a contract (payment per outcome), the agent picks the utility-maximizing action (which maximizes his expected payment from the contract minus his cost of effort), and the principal gets as revenue the expected reward minus payment. 

In our proposed model, there are $n$ different effort levels $i$ (actions). Action $i$ costs $\gamma_i$ units of effort. As before, taking an action triggers a stochastic outcome according to some probability distribution $F_i$ (the $i$-th row of $F$). A new feature relative to \cite{guruganesh2020contracts} is that agent's have a \emph{single-dimensional} private type $c$ --- their cost per unit-of-effort. So their cost for taking action $i$ is $c \cdot \gamma_i$.

The agent's hidden type by \cite{guruganesh2020contracts}, in contrast, is his mapping from actions to outcomes, as modeled by the matrix~$F$ (and is therefore naturally \emph{multi-dimensional}). The cost vector is assumed to be public knowledge. The (discrete) distribution of types is also publicly known. In this model, \cite{guruganesh2020contracts} establish the hardness of computing an optimal truthful contract menu (where optimal refers to the principal's revenue). 
We study a complementary model, where $F$ is known but costs are hidden. The fundamental difference from a computational viewpoint is that in the new model, the type of an agent can be represented by a \emph{single parameter} --- his cost per unit of effort. 
It is well-known from auction design that single-parameter types may allow positive results even when hardness results hold for multi-parameter types. 

Table~\ref{tab:vision} positions our model with respect to the classic work in both mechanism design and contract theory, and on the intersection of the two.

\begin{table*}[t]	
	\begin{center}
        \begin{tabular}{|c|c|c|}
			\toprule
			& Known action & Hidden action \\ \hline
			Known type & Trivial & \makecell{Classic contract theory\\
			\citep{GrossmanHart83}}\\
			&&\\
			\hline
			
			Hidden type & \makecell{Myerson's theory\\
			\citep{Myerson81}} & Our model\\
			(single-dim) & & \\
			\hline
			
			Hidden type & E.g., \citet{CaiDW13} & \cite{guruganesh2020contracts} \\
			(multi-dim) & & \\
			\bottomrule
		\end{tabular}
	\end{center}
	\caption{Our model's relations to other settings in the literature}
	\label{tab:vision}
\end{table*}


\subsection{Our Results}
\label{sub:our-results}

One of the most useful tools that Myerson's theory (see \cite[][]{Myerson81}) provides for single-parameter mechanism design is a characterization of implementable allocation rules. Known as Myerson's lemma, this characterization 
is useful since it defines the design space --- the designer only needs to consider monotone allocation rules,%
\footnote{In the context of mechanism design, an allocation rule maps an agent's (reported) type to his allocation under the mechanism; a payment rule maps his type to what he pays (or in procurement is paid).}
and any such rule is guaranteed to have a corresponding payment rule such that the resulting mechanism is truthful. 

Our main result is a ``Myerson's lemma'' for the more general model of contracts with private costs. Such a lemma should incorporate both the original characterization of Myerson, and the characterization of implementable actions from principal-agent theory of \cite{GrossmanHart83}.
These two previously-known characterizations seem quite different: The first characterizes implementable allocation rules as monotone, which means (in the procurement variant) that agents (service providers)  with high reported costs are not assigned by the mechanism to provide service.%
\footnote{Or, if there are multiple levels of service, agents with high costs are assigned to provide lower, less costly such levels.} The second characterizes implementable actions --- those for which the principal has some contractual payment scheme over the outcomes that incentivizes the agent to choose this action. This characterization says that for an action $i$ to be implementable, there can be no convex combination of the other actions that achieves the same distribution as $i$ over outcomes, at a lower combined cost. 

To contrast with our unified characterization, it's useful to restate monotonicity (in the sense of Myerson's theory) as follows: Assume for simplicity a discrete type space, and consider all agent types and whether they're assigned by the allocation rule to provide service or not. Consider the aggregate 
service by the allocated agents and its overall cost. An allocation rule is monotone (and thus implementable) if and only if there is no combination of agent types that would together provide the same aggregate service at a lower combined cost. 

Our unified characterization result is:


\begin{theorem*}[Informal characterization --- see Theorem~\ref{thm:charac-disc}]
Consider an allocation rule mapping agent types to assigned actions. The rule is implementable if and only if there is no weighted combination of agent types and actions that together achieve the same distribution over outcomes at a lower combined cost.
\end{theorem*}

Our characterization provides a computationally tractable way of checking (in time polynomial in the number of actions, outcomes and agent types) whether a given allocation rule is implementable (see Corollary~\ref{cor:poly-LP}).
For a continuous type space, the characterization is similar in spirit but slightly more involved and appears in Theorem~\ref{thm:char-cont}.

One implication of our characterization is that monotonicity of the allocation rule is no longer sufficient for implementability (as we demonstrate in Proposition~\ref{prop:non-monotone}). Intuitively this happens because the overall cost can now be improved not only by assigning the same actions to agents with lower costs, but also by taking an altogether different combination of actions.

We show two additional applications of our characterization results. The first stands in stark contrast to the APX-hardness of optimal contract design for multi-parameter types demonstrated by \cite{guruganesh2020contracts} (which holds even with only constantly-many actions):

\begin{theorem*}[Tractability for constantly-many actions --- see Theorem~\ref{thm:const-n}]
Finding the optimal contract for single-parameter types is solvable in polynomial time for a constant number of actions.
\end{theorem*}

Our second application (Theorem~\ref{thm:uniform}) uses the simple but powerful observation that if an allocation rule is implementable, the rest of Myerson's theory holds for it despite the general setting of contracts with private types. In particular, Myerson's payment identity holds and the expected revenue is equal to the expected virtual welfare (up to a constant). 
It is thus tempting to consider mechanisms that maximize virtual welfare --- would they turn out to be implementable in our general setting, just like virtual welfare maximizers turn out to be monotone (and thus implementable) in Myerson's original setting? We establish that this is indeed the case when the agent's private cost is distributed uniformly:

\begin{theorem*}[Optimal contract for uniform costs --- see Theorem~\ref{thm:uniform}]
For uniform costs, the virtual welfare maximizing allocation rule is implementable.
\end{theorem*}



We leave as our main open question whether implementability holds for the candidate rule (the optimal monotone rule \`a la Myerson) beyond uniform distributions, and if not what should the mechanism be.

\subsection{Related Work}\label{sec:related-work}


Contract theory is an important and well-studied sub-field of microeconomics with many practical implications; for leading textbooks see \cite{bolton2005contract,salanie,Laffont}. 
Many works on contract design study it entirely separately from mechanism design (screening) --- the former deals with hidden actions of the agent (moral hazard), and the latter with private types of the agent (adverse selection). 
An analysis of the basic single principal, single agent setting (with no types) is found in the seminal work of \citet{GrossmanHart83}, and \citet{Holmstrom80} summarizes some of the classic foundations (see also \cite[]{Nobel16} for the scientific background on the Nobel prize shared by Hart and Holmstr\"om). One of the main take-aways is that the optimal contract can be found in this setting by solving (polynomially-many) linear programs.

A much smaller collection of works studies the combination of moral hazard and adverse selection. 
One classic such work is by \citet{Myerson82}, who studies ``generalized'' principal-agent problems where agents have both private information and ``private decision domains'' (hidden action). The key insight is that the principal may, without loss of generality, restrict herself to incentive compatible direct mechanisms. 
This extends the revelation principle to situations where there are moral hazard factors. We apply this insight in our results. 
\citet{GottliebM13} study a simple setting with an effort/no effort binary choice for the agent and two possible outcomes, so their types are two-dimensional vectors. \citet{GottliebM15} study a more general setting where types and efforts are multi-dimensional (possibly
infinite-dimensional). Their work identifies assumptions under which ``optimal contracts are simple'' in the sense that the optimal (menu of) contracts consists of a single contract for all types. Our work focuses on solving for optimal ``menus'', i.e., contracts with separate payments per type.

\citet{ChadeS19} study a setup similar to ours in which there is a single-parameter type for the agent and the goal is to design the optimal (menu of) contracts. They make different assumptions than us, such as continuous action space or the simplifying MLRP assumption (a strengthening of first-order stochastic dominance among any pair of distributions associated with the actions). They obtain two different sufficient conditions under which they are able to optimally solve the design problem. Their solution involves minimizing the cost of implementing any given action at any given surplus for any given type in a pure moral hazard
problem.
Our focus is on a necessary and sufficient condition for implementability, and on computationally efficient solutions. 

More recent work has started to explore contract design through the computational lens (but without private types)  \cite[][]{babaioff2006combinatorial,HoSV16,DuttingRT19,KR19,DuttingRT20,DEFK21}.

\citet{guruganesh2020contracts} are the first to consider contract design with typed agents from a computational point of view. Importantly, their agent types are multi-dimensional, namely, an agent's type determines the outcome distributions corresponding to every available effort level.
They establish computational hardness of optimal contract design under moral hazard and adverse selection. Their hardness results motivate their exploration of the approximation power of simple classes of contracts (such as linear, fixed-provision contracts).

\citet{CastiglioniM021} study a similar multi-dimensional problem in which a principal seeks to design a single (non-type dependent) contract for a Bayesian agent, whose costs and distributions over outcomes are drawn from known distributions. They establish hardness results for computing the optimal contract, and argue that simple linear contracts provide optimal approximation guarantees subject to poly-time computability.

Further afield, some works consider hidden types of principals rather than agents \cite[e.g.,][]{BernheimW86}.

\section{Model} 
\label{sec:model}

In this section we introduce our model in this paper: a contract design setting with hidden action and a single-parameter private type.
Section~\ref{sub:instance} describes a problem instance, Section~\ref{sub:contract-def} defines contractual solutions, Section~\ref{sub:running-example} presents a running example and Section~\ref{sec:comparison} positions our model relative to classic contract and mechanism design problems. 

\begin{notation}
Let $[n]=\{0,\dots,n\}$ for every $n\in\mathbb{N}$ (i.e., zero is included). 
\end{notation}

\subsection{Single-Parameter Principal-Agent Instance} 
\label{sub:instance}

An instance (a.k.a.~setting) of our model consists of two players, a \emph{principal} and an \emph{agent}. An \emph{action} set~$[n]$ is available to the agent. The agent's chosen action leads to an \emph{outcome} $j\in [m]$, with \emph{reward} $r_j\geq 0$ for the principal. We assume without loss of generality that outcomes are ordered in a non-decreasing order by their \emph{rewards} $r_0\leq...\leq r_m$, and that there is an outcome with no reward for the principal, i.e., $r_0=0$. In what follows we do not distinguish between the outcomes and their rewards. 

Action $i\in[n]$ requires $\gamma_i \ge 0$ \emph{units of effort} from the agent, and induces a distribution (probability mass function) $F_{i}$ over the $m$ outcomes/rewards. Let $F_{i,j}$ denote the probability of outcome~$j$ when the agent takes action $i$. We assume that only the first action requires no effort from the agent, i.e., $\gamma_0=0$, and that actions are ordered by the amount of effort they require, i.e., $\gamma_0< ...\leq \gamma_n$. We also assume that the first outcome occurs if and only if the agent takes the first action, i.e., $F_{0,0}=1$ and $F_{i,0}=0$ for $ 1\leq i\leq n$. Thus, the principal can monitor whether or not the agent takes the zero-cost action. This assumption is a simple way to model the agent's opportunity to opt out of the contract when individual rationality (i.e., guarantee for non-negative utility) is not satisfied \cite[e.g.,][]{HoSV16}.  

Let 
\begin{equation}
R_i=\mathbb{E}_{j\sim F_i}[r_j]=\sum_{j\in [m]} F_{i,j} r_j
\label{eq:expected-reward}
\end{equation}
be the \emph{expected reward} of action $i\in[n]$. We assume (as in \cite[][]{DuttingRT19}) that there are no ``dominated'' actions: every two actions $i<i'$ have distinct expected rewards $R_i\ne R_{i'}$, and the action that requires more units of effort has the higher expected reward,  i.e., $R_{i}<R_{i'}$. Thus actions are also ordered by expected reward $0=R_0<\dots< R_n$. 

\paragraph{Agent's type.} The agent has a single-parameter \emph{type} $c$. The type is drawn from a distribution~$G$ supported over the set of all possible types $C\subseteq \mathbb{R}_{\ge 0}$.
An agent's type captures his \emph{cost per unit-of-effort}, or in other words, his \emph{marginal} cost for effort. 
When an agent of type $c$ takes action $i$, the principal gains a reward $r_j$ drawn according to distribution $F_i$, and the agent loses $\gamma_i c$ (the number of effort units that action $i$ requires multiplied by the agent's cost per unit). 

We distinguish between instances with a \emph{discrete} type space $C$ and those with a \emph{continuous} one. We denote the former by $\Dis(F,\gamma,r,G,C)$, {in which case $G$ is a discrete density function}. We denote the latter by $\Con(F,\gamma,r,G,C)$, in which case $g$ denotes the probability density function, and $G$ denotes the cumulative distribution function. When clear from the context, we sometimes omit some parameters from $\Dis$ and $\Con$. 
For the continuous case we assume (similarly to~\cite{Myerson81}) that $C=[0,\bar{c}]$ for $0<\bar{c}<\infty$.%
\footnote{Our results hold more generally for $C=[\underline{c},\bar{c}]\subseteq \mathbb{R}_{\ge 0}$.} 


\paragraph{\bf Summary of an instance and who knows what.}
To summarize, an instance is described by
distributions $F=(F_0,\dots,F_n)$ and corresponding ``effort levels'' $\gamma=(\gamma_0,\dots,\gamma_n)$ for the actions,
rewards $r=(r_0,\dots,r_m)$ for the outcomes, and a type distribution $G$ over support $C$ for the agent. 
We omit from the notation components of the setting that are clear from the context.

The instance itself is publicly known.
The action $i$ which the agent actually takes is \emph{hidden} from the principal, who only observes its stochastic reward $r_j$. The agent's realized type $c$ is \emph{privately-known} only to the agent himself.

\subsection{Contracts}
\label{sub:contract-def}

Our notion of a contract is a generalization of the standard one (see, e.g., \cite[][]{DuttingRT19}). The generalization is to accommodate for agent types, and is the direct revelation version of the ``menu of contracts'' notion studied by~\cite{guruganesh2020contracts}.
More formally, a \emph{contract} $(x,t)$ is composed of an \emph{allocation rule $x:C\to[n]$} and a \emph{payment rule $t:C\to \mathbb{R}^{m+1}_{\ge 0}$}. 
A type report $c' \in C$ is solicited from the agent (where $c'$ may differ from the true type $c$), and the allocation rule maps $c'$ to an action 
$x(c')$. 
An important difference from mechanism design is that the action can only be \emph{recommended} to the agent by the contract, rather than \emph{enforced} like an allocation by an auction. 
Whether or not the agent adopts the recommendation depends on the payment rule. A payment rule $t$ maps $c'$ to $m$~non-negative payments or \emph{transfers} $(t^{c'}_0,...,t^{c'}_m)$, one for each outcome. 
The transfers are associated with outcomes rather than actions since the actions are hidden from the principal.
For action $i\in[n]$ let
\begin{equation}
T^{c'}_i=\mathbb{E}_{j\in F_i}[t^{c'}_j]=\sum_{j\in [m]} F_{i,j} t^{c'}_j\label{eq:expected-transfer}
\end{equation}
denote the expected payment from principal to agent with reported type $c'$ for taking action $i$.
All transfers are required to be non-negative; this guarantees the standard \emph{limited liability} property for the agent, who is never required to pay out-of-pocket (see, e.g., \cite[][]{innes1990limited, gollier1997risk,carroll2015robustness}).%

\begin{remark}
It is possible for the allocation rule $x$ to be randomized in the sense of mapping type~$c'$ to a distribution over actions. 
In Appendix~\ref{appx:rand} we show that such randomized allocation rules have no extra power in the settings we consider. Thus we restrict attention to deterministic rules unless stated otherwise. More complex contract formats in which the randomization is also over payment rules are beyond the scope of our work, and they constitute an interesting avenue for future research.  Such contracts are formally presented in Appendix~\ref{appx:rand}, along with a short discussion of their power.

\end{remark}

\paragraph{\bf The game.}
A contract $(x,t)$ induces the following two-stage game:
\begin{itemize}[leftmargin=0.63in]
    \item[Stage 1:] The agent submits a type report $c'$, fixing the contractual payment vector $t^{c'}=(t^{c'}_0,...,t^{c'}_m)$.
    \item[Stage 2:] The agent chooses an action $i$ (which is not necessarily the action $x(c')$ prescribed to it by the contract), and incurs a cost of $\gamma_ic$ (where~$c$ is his true type). An outcome $j$ is realized according to distribution $F_i$. The principal is rewarded~$r_j$ and pays $t^{c'}_j$ to the agent. 
\end{itemize}

\begin{remark}
Once the game reaches Stage 2, we are back in a standard principal-agent setting with no types, in which the agent faces a standard contract (simply a vector of $m$ payments), and chooses accordingly a costly action that rewards the principal. 
\end{remark}

\paragraph{\bf Utilities.} 

Let $c$ be the true type, $c'$ the reported type, $i$ the action chosen by the agent and $j$ the realized outcome. The players' utilities are as follows: $t^{c'}_j-\gamma_i c$ for the agent, and $r_j-t^{c'}_j$ for the principal. In expectation over the random outcome $j\sim F_i$ these are $T^{c'}_i-\gamma_i c$ for the agent and $R_i-T^{c'}_i$ for the principal. The sum of the players' expected utilities is the expected \emph{welfare} from action $i$, namely $R_i-\gamma_i c$.

Let us now consider the agent's rational behavior given his expected utility. When facing payment vector $t^{c'}$ in Stage~2, an agent whose true type is $c$ will choose the action $i^*(c',c)$ that maximizes his expected utility:
$$
i^*(c',c)\in \arg\max_{i\in [n]}\{T^{c'}_i-\gamma_i c\}.
$$
As is standard in the contract design literature, if there are several actions with the same maximum expected utility for the agent, we assume consistent tie-breaking in favor of the principal (see, e.g., \cite[][]{guruganesh2020contracts}).
Thus $i^*(c',c)$ is well-defined.
%
When reporting his type in Stage~1, the agent will report $c'$ that maximizes $T^{c'}_{i^*(c',c)}-\gamma_{i^*(c',c)} c$, that is, his expected utility given his anticipated choice of action in Stage~2. 

\paragraph{\bf Incentive compatibility.}

We say that a contract $(x,t)$ is \emph{incentive compatible (IC)} if for every type $c\in C$, it is in the best interest of an agent of type $c$ to both truthfully report $c$ in Stage~1 and to take the prescribed action $x(c)$ in Stage~2. Formally:

\begin{definition}
A contract $(x,t)$ is \emph{IC} if \;$\forall c\in C$, 
\begin{enumerate}
    \item $x(c)=i^*(c,c)$; and
    \item $c\in \arg\max_{c'} \{T^{c'}_{i^*(c',c)}-\gamma_{i^*(c',c)} c\}$. 
\end{enumerate}
\end{definition}
\noindent
Condition (1) in the above definition ensures that the prescribed action for type $c$ maximizes the agent's expected utility when he truthfully reports his type~$c$.%
\footnote{The contracts we consider are designed in favor of the principal and so the action recommended by $x$ can be compatible with the tie-breaking rule determining $i^*$.}
Condition (2) ensures that reporting truthfully is a (weakly) dominant strategy.

Note that focusing on IC contracts is without loss of generality by the revelation principal (see \cite[][]{Myerson82}). 
Also, incentive compatibility usually goes hand in hand with individual rationality (IR), which in our context requires that the utility of an agent who reports truthfully 
is non-negative. 
Here the assumption that $\gamma_0=0$ comes in handy, as it means that the agent can always choose an action that requires zero effort. Together with limited liability this ensures non-negative utility. 

We mention that classic principal-agent theory also assumes individual rationality for the principal. In our model, the principal can always guarantee herself non-negative utility by paying zero for all actions, incentivizing no-effort for all types.

\paragraph{\bf Objective.} 
The principal's goal is to design an \emph{optimal} contract $(x,t)$, i.e., a contract that satisfies IC and maximizes her expected utility, where the expectation is over the agent's random type $c$ drawn from $G$, as well as over the random outcome of the agent's prescribed action $x(c)$:
$$
\mathbb{E}_{c\sim G}[R_{x(c)}-T^{c}_{x(c)}] = \mathbb{E}_{c \sim G} \left[\mathbb{E}_{j \sim F_{x(c)}}[r_j - t^c_j] \right].
$$

\subsection{An Example}
\label{sub:running-example}

The following principal-agent instance will serve as the running example of this paper.

\begin{example}\label{ex:runing}
There are two agent types, four actions with required effort levels $\gamma_0=0$, $\gamma_1=1$, $\gamma_2=3$, $\gamma_3=10$, and three outcomes with rewards $r_1=0$, $r_2=10$, $r_3=30$ (that is, $n=m=3$). The distributions over outcomes are $F_0=(1,0,0)$, $F_1=(0,1,0)$, $F_2=(0,0.5,0.5)$, and $F_3=(0,0,1)$. The two types are $c = 1$ and $c = 4$, and they occur with equal probability.
\end{example}


A possible contract for Example~\ref{ex:runing} is $x(1) = 3$ with payments $t^1 = (0,0,14)$, and $x(4) = 1$ with payments $t^4 = (0,4,0)$, as depicted in Figure~\ref{fig:ex}. Intuitively, the ``stronger'' type who incurs less cost per unit-of-effort ($1$ rather than $4$) is recommended a more strenuous action (the fourth rather than second action).

\begin{figure}[h!]
\begin{minipage}{0.49\textwidth}
\centering
\begin{tabular}{lll|l}
\toprule
$t_1(1) = 0$ & $t_2(1) = 0$ & $t_3(1) = 14$\\
$r_1 = 0$ & $r_2 = 10$ & $r_3 = 30$ \\
\midrule
1 & 0 & 0 & $\gamma_0 = 0$ \\
0 & 1 & 0 & $\gamma_1 = 1$  \\
0 & 0.5 & 0.5 & $\gamma_2 = 3$ \\
0 & 0 & 1 & $\gamma_3 = 10$\\
\bottomrule
\end{tabular}
\end{minipage}
\begin{minipage}{0.49\textwidth}
\centering
\begin{tabular}{lll|l}
\toprule
$t_1(4) = 0$ & $t_2(4) = 4$ & $t_3(4) = 0$\\
$r_1 = 0$ & $r_2 = 10$ & $r_3 = 30$ \\
\midrule
1 & 0 & 0 & $\gamma_0 = 0$ \\
0 & 1 & 0 & $\gamma_1 = 1$  \\
0 & 0.5 & 0.5 & $\gamma_2 = 3$ \\
0 & 0 & 1 & $\gamma_3 = 10$\\
\bottomrule
\end{tabular}
\end{minipage}
\caption{A contract for Example~\ref{sub:running-example}. The left tableau shows the payments for reported type $c' = 1$ (the recommended action is the fourth one), and the right tableau those for reported type $c' = 4$ (the recommended action is the second one).}\label{fig:ex}
\end{figure}

For this contract to be IC, we would like the agent to choose the tableau (read payments) that corresponds to his true type, as well as the recommended action for that type. For an agent with type $c = 1$ this means choosing the left tableau and fourth action. This gives him expected utility of $1 \cdot 14 - 10 \cdot 1 = 4$. An agent with type $c = 4$ should choose the right tableau and second action. This choice yields expected utility of $1 \cdot 4 - 1 \cdot 4 = 0$.
For neither type should the agent wish to pretend that he is of a different type and/or choose a different action. So, for example, the agent with type $c = 1$ should not be incentivized to pretend to be of type $c' = 4$ and take action $1$. Indeed, in this case his expected utility would be $1 \cdot 4 - 1 \cdot 1 = 3$, which is smaller than the expected utility he gets for truthfulness.
This is in fact true for all types and possible deviations, and so this contract is indeed IC. The principal's expected utility is $0.5 \cdot (30-14) + 0.5 \cdot (10-4) = 11$. 


\subsection{Relation to Classic Settings}
\label{sec:comparison}

Consider a simple instance of our model in which there is only a single type in $C$, so that the type of the agent is publicly known. For such instances our model reduces to the classic principal-agent model of \citet{GrossmanHart83}.%
\footnote{As in \cite[][]{carroll2015robustness}, we couple the classic model with the particular form of agent risk-aversion captured by limited liability.}
The goal in that model is to design a standard contract, composed of a single payment vector and an (implicit) action recommendation, such that the agent takes the recommended action and the expected utility of the principal is maximized.

At the other extreme, consider an instance of our model in which every action $i$ deterministically leads to a distinct reward $r_i$ (i.e., the distributions $F$ are point mass; without loss of generality, when viewed as a tableau as in Figure~\ref{fig:ex} they constitute the identity matrix). In this case, the agent's action is not hidden. For such instances our model reduces to a reverse (procurement) variant of the single-parameter mechanism design setting of \citet{Myerson81}. In this variant, the agent acts as a seller of a service whose cost for providing the service is private. The service can be supplied at different (known) levels $\gamma_0,\dots,\gamma_n$, and the values $\{r_i\}_i$ of the principal/buyer for these levels are publicly known. The goal is to design a truthful procurement auction that maximizes the buyer's expected revenue. 
The special case in which there are two actions ($n=1$) corresponds to the standard Bayesian procurement setting in which the auction's outcome is whether or not to buy from the agent/seller depending on his declared cost (drawn from a known distribution). In this special case, a decision to buy corresponds to a recommendation of the costly action ($x(c)=1$), and not buying corresponds to the zero-cost action ($x(c)=0$).

\section{Characterization of Contract Implementability}
\label{sec:charac}

A main driver of mechanism design with single-dimensional types has been Myerson's theory (see \cite[][]{Myerson81}). It  characterizes the type of auction allocation rules that can be turned into incentive compatible auctions --- i.e., that are \emph{implementable} --- as those which are monotone. Furthermore, for each implementable allocation rule it provides an essentially unique payment rule that turns this rule into an IC mechanism.\footnote{Technically, \citet{Myerson81} only considered the continuous type case. A similar characterization, however, applies also in the discrete type case \cite[e.g.,][]{BergemannP07,Elkind07}. Allocation rules still have to be monotone, and while the payment identity no longer applies, it can be shown that payments have to be in a certain range, and hence there is still a revenue-optimal choice for any given allocation rule.} 
In this section we develop such a theory for the single-parameter principal-agent model with hidden action and private cost. The main question is to find necessary and sufficient conditions for an allocation rule to be implementable as an IC contract. 


\begin{definition}
\label{def:implementable}
An allocation rule $x:C\to[n]$ is \emph{implementable} if there exists a payment rule $t:C\to \mathbb{R}^{m+1}_{\ge 0}$ such that contract $(x,t)$ is IC. 
\end{definition}

In Definition~\ref{def:implementable} the payment vector $t$ is required to be non-negative, thus imposing limited liability.
We note that this requirement is without loss of generality since every allocation rule that is implementable with arbitrary, possibly negative payments is also implementable with non-negative payments (by adding some offset to all payments). So, in principle, equivalent characterizations can be obtained without requiring payments to be non-negative. For our results in Section~\ref{sec:applications}, however, the constant offset and its interplay with limited liability will play a role, and for this it will be useful to develop the general machinery while requiring non-negative payments. 

In Section~\ref{sec:properties} we define relevant properties of allocation rules, in particular monotonicity in the context of contracts. Sections~\ref{sec:discrete} and \ref{sec:continuous} give our characterization for discrete and continuous types, respectively, and show monotonicity is necessary but not sufficient for implementability. Throughout this section and unless stated otherwise, by an ``allocation rule'' we refer to such a rule in the context of contracts.

\subsection{Properties of Allocation Rules} \label{sec:properties}

Two properties of allocation rules that will play a role in our characterization results are monotonicity and the special case of piecewise constant monotonicity.

A monotone allocation rule recommends actions that are (weakly) more costly in terms of effort --- and hence also more rewarding in expectation for the principal --- to agents whose cost per unit-of-effort is lower. Intuitively, such agents are better-suited to take on effort-intensive tasks.

\begin{definition}\label{def:monotone}
    An allocation rule $x:C\to[n]$ is \emph{monotone} if for every $c,c'\in C$, 
    $$
    c<c'\implies x(c)\ge x(c').
    $$
\end{definition}


A special class of monotone allocation rules are piecewise constant allocation rules (see Figure~\ref{fig:Fig2}). Informally, these are allocation rules such that: (i) the allocation function $x(\cdot)$ is locally constant in distinct intervals of $[0,\bar{c})$; and (ii) the allocation is decreasing with intervals. 


\begin{definition}
\label{def:piecewise-constant}
    An allocation rule $x:C\to[n]$ is \emph{monotone piecewise constant} if there exist $\ell+2$ \emph{breakpoints} $0=z_0< ...< z_{\ell}< z_{\ell+1}=\bar{c}$ where $\ell\leq n$, such that $x(z_{i})\geq x(z_{i+1})$ $\forall i\in [\ell]$ and for every $c\in C$,
    \begin{eqnarray*}
    c \in (z_{i},z_{i+1}) \implies x(c)=x(z_i);
    \end{eqnarray*}
    and without loss of generality $x(\bar{c})=x(z_\ell)$.%
    \footnote{Note that in the continuous model, changing the allocation $x(\cdot)$ for a finite number of types does not change the objective. This is also the reason that the allocation of an interval $(z_i,z_{i+1})$ can be imposed without loss of generality on its left endpoint $z_i$ but not on its right one $z_{i+1}$.
    }
\end{definition}


Note that in our setting, since the image of $x$ is finite, every monotone allocation rule can be viewed as corresponding to a monotone piecewise constant function. 

 
 
 

\subsection{Characterization for Discrete Types}
\label{sec:discrete}

We now present our characterization of implementable allocation rules for discrete types. 
Our approach must encompass the standard LP-based argument in contract theory for establishing whether or not a given action is implementable (see, e.g., Appendix A.2 of~\citep{DuttingRT19}).%
\footnote{In the classic principal-agent model with no private type, an action is \emph{implementable} if there exists a payment rule such that this action maximizes the agent's expected utility.}
At the same time, it must generalize Myerson's characterization of implementable allocation rules in single-parameter mechanism design problems (for discrete types). Our characterization thus shares features of Myerson's theory but also departs from it in significant ways.

\begin{theorem}[Discrete Characterization]\label{thm:charac-disc} 
An allocation rule $x: C \rightarrow [n]$ is implementable if and only if there exist no weights $\lambda_{(c,c',k)} \geq 0$ for pairs of types $c,c' \in C$ and actions $k \in [n]$ which satisfy $\sum_{c'\in C}\sum_{k\in [n]}\lambda_{(c,c',k)}=1$ $\forall c\in C$ and the following conditions: 
\begin{enumerate}
    \item {\sc Weakly dominant distributions:}\label{itm:distrib-dicrete-char}
    \begin{eqnarray*}
    \sum_{{c'\in C}}\sum_{k\in [n]}  F_{k,j} \lambda_{(c',c,k)} \geq F_{x(c),j}  &\forall c \in C,j \in [m].
    \end{eqnarray*}
    \item {\sc Strictly lower joint cost:}\label{itm:costs-dicrete-char}
    \begin{eqnarray*}
    \sum_{c,c'\in C}\sum_{k\in [n]} \lambda_{(c,c',k)}\gamma_k c <
    \sum_{c\in C} \gamma_{x(c)} c.
    \end{eqnarray*}
\end{enumerate}
\end{theorem}

One can think of the weights $\lambda_{(c,c',k)}$ in Theorem~\ref{thm:charac-disc} as a \emph{deviation plan}, by which whenever the true type of the agent is $c$ then with probability $\lambda_{(c,c',k)}$ he deviates to type $c'$ and to action~$k$ (i.e., he pretends to be of type $c'$ and takes action $k$).\footnote{Note that $\lambda_{(c,c,x(c))}$ might be strictly positive. In this case, part of type $c$'s deviation plan is to report truthfully and take the allocated action.} Then, the first condition says that for every type~$c$, the combination of distributions over outcomes resulting from deviations to $c$ dominates the distribution of a non-deviating agent whose true type is $c$.
{The second condition compares the total expected cost of the deviation plan to the total cost of a non-deviating, truthful agent; it says that in summation over all agent types, the total cost of the deviation plan is strictly lower than that of truthfulness.}
Taking the two conditions together, Theorem~\ref{thm:charac-disc} states that an allocation rule is implementable precisely when there is no deviation plan for the agent which costs less than being truthful, and enables the agent to dominate the distribution over outcomes achievable by a truthful agent of any type (by pretending to be of that type).

\begin{proof}[Proof of Theorem \ref{thm:charac-disc}]

We take the following LP-based approach to determine whether or not $x$ is implementable. The LP for finding an IC contract that implements $x$ has $|C|m$ payment variables $\{t^c_j\}$ which must be non-negative by limited liability, 
and $|C|^2(n+1)$ constraints ensuring that type~$c$'s expected utility from action $x(c)$ given contract $t^c$ is at least the expected utility from any other action $k\in [n]$ given any other contract $t^{c'}$. 
Recall from Eq.~\eqref{eq:expected-transfer} that $T^{c}_{k}=\sum_{j=1}^{m} F_{k,j} t^{c}_j$ is the expected transfer to an agent for reporting type $c$ and taking action $k$.
The LP is:
\begin{equation*}
\begin{array}{ll@{}ll}\tag{LP1}\label{LP1}
\text{max} & 0 &\\
\text{s.t.} & T_{x(c)}^{c} -\gamma_{x(c)} c \geq T^{c'}_{k} -\gamma_{k} c & &\forall c,c' \in C,k \in [n],\\
& t^{c}_j\geq 0 & & \forall c\in C, j\in [m].
\end{array}
\end{equation*}
The dual of the above linear program has $C^2 (n+1)$ non-negative variables $\lambda_{(c,c',k)}$ indexed by $c,c'\in C$ and $k\in [n]$, and $Cm$ constraints.
The dual is:
\begin{equation*}
\begin{array}{ll@{}ll}
\tag{D1}
\label{Dual1}
\text{min} & \sum_{c,c'\in C}\sum_{k\in [n]} \lambda_{(c,c',k)} (\gamma_k-\gamma_{x(c)})c \text{\space} &\\
\text{s.t.}&   \sum_{c'\in C,k\in [n]}F_{k,j} \lambda_{(c',c,k)} \geq
F_{x(c),j} (\sum_{c''\in C,k\in [n]} \lambda_{(c,c'',k)}) & &  \forall c \in C,j\in [m],  \\
&  \lambda_{(c,c',k)}\geq 0 & &  \forall c,c'\in C, k\in [n].
\end{array}
\end{equation*}
By strong duality, \ref{LP1} is feasible (i.e., $x$ is implementable) if and only if its dual has an optimal solution with objective value of zero. 
Notice that the dual is feasible (e.g., setting $\lambda_{(c,c,x(c))}=1$ for every type $c$ and all other $\lambda$s to zero yields a feasible solution).
Thus \ref{LP1} is \emph{not} feasible if and only if there exists a feasible solution $\lambda$ to the dual which has objective value strictly below zero. Lemma~\ref{lemma:normalize-dis-char} in Appendix~\ref{appx:Char} shows that the existence of $\lambda$ implies the existence of a feasible solution with strictly negative objective $\lambda'$ in which $\sum_{c''\in C}\sum_{k\in [n]}\lambda'_{(c,c'',k)}=1$ $\forall c\in C$. It can be verified that when $\sum_{c''\in C}\sum_{k\in [n]}\lambda'_{(c,c'',k)}=1$ $\forall c\in C$, strictly negative objective value simplifies to 
$\sum_{c,c'\in C}\sum_{k\in [n]} \lambda'_{(c,c',k)} \gamma_k c<\sum_{c\in C}\gamma_{x(c)}c,$ and the first set of constraints simplifies to $\sum_{c'\in C}\sum_{k\in [n]}F_{k,j} \lambda'_{(c',c,k)}\geq F_{x(c),j}$ $\forall c\in C, j\in [m].$ This completes the proof.
\end{proof}

In Appendix~\ref{appx:relation-lit} we elaborate on the relation between Theorem~\ref{thm:charac-disc} and classic characterizations. Specifically, we show how it reduces to the classic characterizations when, as specified in Section~\ref{sec:comparison}, there is either no hidden type (standard contract setting) or no hidden action (standard procurement setting).

\paragraph{\bf Computational implications.} 

An immediate corollary of the LP-based proof of the characterization is that we can tractably: (i)~verify whether a given allocation rule is implementable; and (ii)~compute optimal payments for this allocation rule if it is.
For (i) we verify whether \ref{LP1} has a feasible solution, and for (ii) we minimize the expected payment $T_{x(c)}^c= \sum_{j} F_{x(c),j} t^c_j$ over the feasibility region of \ref{LP1}. 

\begin{corollary}
\label{cor:poly-LP}
For a discrete type space $C$, the problem of determining whether or not an allocation rule is implementable is solvable in time $O(poly(n,m,|C|^2))$. Moreover,  if an allocation rule is implementable, then we can find optimal payments for this allocation rule in time $O(poly(n,m,|C|^2))$.
\end{corollary}

We remark that if the type space is large ($|C|\gg n$) but the allocation rule is represented succinctly (e.g.,~by specifying which ranges of types map to which of the $n$ actions), the running time stays polynomial in $n,m$ (see Corollary~\ref{cor:poly-LP-cont} for the extreme case where $C$ is infinite).

\paragraph{\bf Monotonicity is necessary but not sufficient.} 

We next observe that beyond the two extremes with just hidden type or just hidden action, the conditions in Theorem~\ref{thm:charac-disc} imply that any implementable allocation rule for a setting with private type and hidden action has to be monotone in the sense of Definition~\ref{def:monotone}.

\begin{proposition}
\label{prop:impl-is-mon-dis}
For discrete types, every implementable allocation rule is monotone. 
\end{proposition}

The proof idea is simple --- 
consider a ``swap'' between two types $\ell$ and $h$ that certify a violation of monotonicity, and use it to derive a deviation plan that satisfies the non-implementability conditions of Theorem~\ref{thm:charac-disc}.

\begin{proof}[Proof of Proposition~\ref{prop:impl-is-mon-dis}]
Suppose towards contradiction that there exists an implementable allocation rule~$x:C\to [n]$ which is non-monotone. By definition, there are two types $\ell<h$ such that $\gamma_{x(\ell)}<\gamma_{x(h)}$. We specify weights (a deviation plan) that satisfy the conditions for non-implementability in Theorem~\ref{thm:charac-disc}, contradicting our assumption that $x$ is implementable. 
Let $\lambda_{(\ell,h,x(h))}=\lambda_{(h,\ell,x(\ell))}=1$, 
and $\lambda_{(c,c,x(c))}=1$ $\forall c \in C \setminus \{h,\ell\}$, and $\lambda_{(c,c',k)}=0$ for all other entries. Simply put, $\lambda$ swaps the allocations of types $\ell$ and $h$.

We first show that $\lambda$ gives weakly dominant distributions (Condition~\eqref{itm:distrib-dicrete-char} in Theorem~\ref{thm:charac-disc}). For $c\notin \{\ell,h\}$, the inequality 
$\sum_{{c'\in C}}\sum_{k\in [n]}  F_{k,j} \lambda_{(c',c,k)} \geq F_{x(c),j}$ holds with equality $\forall j\in [m]$ by the definition of $\lambda$. For $c=\ell$, since the only non-zero $\lambda$ on the left-hand side is $\lambda_{(h,\ell,x(\ell))}=1$, the left-hand side is equal to $F_{x(\ell),j}$. Thus, the inequality holds with equality for $c=\ell$. The same arguments apply for $c=h$. 

Next, we show strictly lower joint cost, i.e., $\sum_{c,c'\in C}\sum_{k\in [n]} \lambda_{(c,c',k)}\gamma_k c < \sum_{c\in C} \gamma_{x(c)} c$.
By definition of $\lambda$, this simplifies to $\gamma_{x(\ell)} h+\gamma_{x(h)} \ell<\gamma_{x(\ell)} \ell+\gamma_{x(h)} h$. It can be verified that the last inequality holds since $\ell<h$ and $\gamma_{x(\ell)}<\gamma_{x(h)}$. This completes the proof.
\end{proof}

Interestingly, however, we show in Proposition~\ref{prop:non-monotone} below that monotonicity is \emph{not} a sufficient condition for implementability in our model with private type and hidden action --- a clear distinction from Myerson's theory \citep{Myerson81}.

This distinction is not caused just by the fact that in contract theory sometimes certain actions are not implementable at all 
(at least not in a way that ensures the principal non-negative utility). 
In fact, our running example that we use to establish Proposition~\ref{prop:non-monotone} satisfies a condition for the implementability of any action  (see \citep[Appendix A.2]{DuttingRT19}).
That is, for every type $c$ and every action, there is a payment vector that incentivizes an agent of type $c$ to take that action, but there is no payment rule that incentives truthful reporting of every type while also ensuring that 
the agent prefers to take the prescribed action for that type. 

\begin{proposition}\label{prop:non-monotone}
Consider Example~\ref{ex:runing}. The monotone allocation rule $x(1)=x(4)=2$ is unimplementable.
\end{proposition}

\begin{proof}
The proof relies on the characterization in Theorem~\ref{thm:charac-disc}. 
Suppose towards a contradiction that $x$ is implementable. We specify weights that satisfy the conditions for unimplementability in Theorem~\ref{thm:charac-disc}.
Let $\lambda_{(1,1,3)}=\lambda_{(1,4,3)}=\lambda_{(4,1,1)}=\lambda_{(4,4,1)}=0.5$, and $\lambda_{(c,c',k)}=0$ otherwise. One may interpret $\lambda$ as the deviation where type $1$ takes action $3$, type $4$ takes action $1$, and both types report either $c'=1$ or $c'=4$ with equal probability. 

We first show that $\lambda$ gives weakly dominant distributions (Condition~\eqref{itm:distrib-dicrete-char} in Theorem~\ref{thm:charac-disc}). 
For $c=1$, since the only non-zero $\lambda$s on the left-hand side are $\lambda_{(4,1,1
)}=\lambda_{(1,1,3)}=0.5$, the left-hand side (the generated distribution over outcomes) is equal to $F_{2,j}=F_{x(1),j}$. Thus, the inequality holds with equality for $c=1$. The same arguments apply for $c=4$. 

Next, we show $\sum_{c,c'\in C}\sum_{k\in [n]} \lambda_{(c,c',k)}\gamma_k c < \sum_{c\in C} \gamma_{x(c)} c$, i.e., strictly lower joint cost (Condition~\eqref{itm:costs-dicrete-char} in Theorem~\ref{thm:charac-disc}).
By definition of $\lambda$ and recalling that $\gamma_1=1,\gamma_2=3,\gamma_3=10$, this simplifies to $14=\gamma_{1}\cdot 4+\gamma_{3}\cdot 1<\gamma_{2}\cdot 1+\gamma_{2}\cdot 4=15$.
That is, the initial joint cost is $15$, while the joint cost after deviation is $14$. This completes the proof.
\end{proof}

We note that while at first glance Proposition~\ref{prop:non-monotone} seems to indicate an inherent incompatibility between monotonicity and implementability, it could still be the case that the \emph{optimal} monotone allocation rule is always implementable. We will return to this question in Section~\ref{sec:applications}.

\subsection{Characterization for Continuous Types}
\label{sec:continuous}

We next tackle the case of continuous types. We establish our characterization result (Theorem~\ref{thm:char-cont}) by reduction to the discrete case with an extra condition (Lemma~\ref{lemma:reduction}). The key idea is as follows: We can argue --- independently of the LP approach to characterization --- that any implementable allocation rule must be monotone (Proposition~\ref{prop:impl-is-mon-cont}). Because the image of $x$ is finite, this means $x$ must be monotone piecewise constant (see Figure~\ref{fig:Fig2} on the left). The characterization can then be phrased for the $O(n)$ breakpoints $\{z_i\}$ of $x$. 
%
%
%
Like in Theorem~\ref{thm:charac-disc}, the characterization states that there should be no ``profitable'' deviation plan, where this time deviations refer to breakpoints~$\{z_i\}$. 
For a technical reason explained below, the deviation plan involves two copies of every breakpoint 
that we refer to as $R$ (Right) and $L$ (Left).



\begin{theorem}[Continuous Characterization]
\label{thm:char-cont}
An allocation rule $x: \mathbb{R} \rightarrow [n]$ is implementable if and only if it is monotone piecewise constant and there exist no weights $\lambda_{(R,i,{i'},k)},\lambda_{(L,i+1,{i'},k)}\geq 0,$ for all $i\in [\ell],i'\in [\ell+1],k \in [n]$ which satisfy $\sum_{i',k\in [n]}\lambda_{(L,i+1,i',k)}=\sum_{i',k\in [n]}\lambda_{(R,i,i',k)}=1$ $\forall i\in [\ell]$ and the following conditions:
\begin{enumerate}
    \item {\sc Weakly dominant distributions:}
    \begin{eqnarray*}
    \sum_{i'\in [\ell],k\in [n]}\frac{1}{2}(\lambda_{(R,i',i,k)}F_{k,j}+ \lambda
_{(L,i'+1,i,k)}F_{k,j})\geq 
F_{x(z_i),j} && \forall  j \in [m], i\in [\ell].
    \end{eqnarray*}
    \item {\sc Strictly lower joint cost:}\label{item:cont-dist-const} 
    \begin{eqnarray*}
    \sum_{i\in [\ell]}(\gamma_{x(z_i)} z_i+\gamma_{x(z_{i})} z_{i+1}) >\sum_{i\in [\ell],i'\in [\ell+1],k\in [n]}\lambda_{(R,i,i',k)}\gamma_{k} z_i+\lambda_{(L,i+1,i',k)}\gamma_{k} z_{i+1}.
    \end{eqnarray*}
\end{enumerate}
\end{theorem}





\paragraph{\bf Our discretization.}

The reason we index the $\lambda$s with $R$ or $L$ is 
as follows. The set of breakpoints $z_0,...,z_{\ell+1}$ of the monotone piecewise constant allocation rule $x$ form a natural continuous-to-discrete reduction. However, ensuring only that type $z_i$ has no incentive to deviate from $x(z_i)$ is not sufficient, since this does not maintain IC for non-breakpoints (i.e., types inside the interval $(z_i,z_{i+1})$). 
A helpful observation is that ensuring IC for both ``endpoints'' of $[z_i,z_{i+1})$ \emph{does} suffice to guarantee that all types in the interval do not deviate from $x(z_i)$, as these types share the same allocation and related incentives. We thus introduce the following discretization: We ``split'' the $i$th breakpoint (save $z_0$ and $z_{\ell+1}$) into two discrete types, $(L,i)$ and $(R,i)$. Then, consider an interval $[z_{i-1},z_i)$. For the left breakpoint, with cost $z_{i-1}$, we have type $(R,i-1)$, and for the right breakpoint with cost (that tends from the right to) $z_{i}$ we have type $(L,i)$.


\begin{figure}[h]
   \centering
    \captionsetup{justification=centering,margin=1.5cm}
    \includegraphics[scale=1,width=140mm]{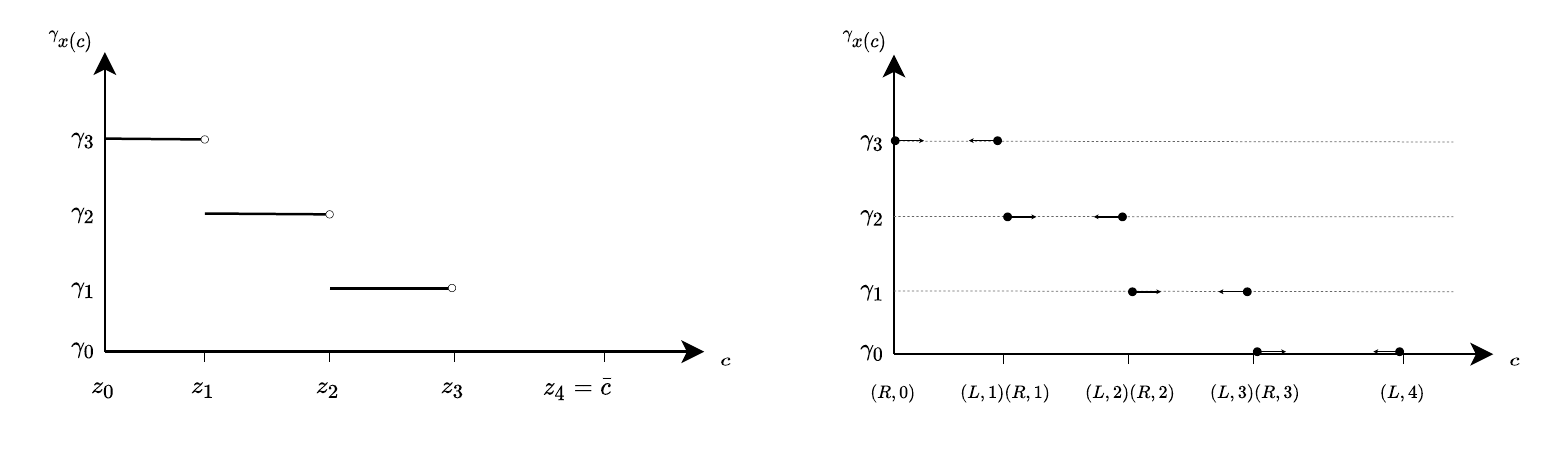}
    \caption{{\bf Left}: A monotone piecewise constant allocation rule $x(\cdot)$ for $4$ actions; the $x$-axis is the continuous type space and the $y$-axis is $\gamma_{x(c)}$. \\{\bf Right}: The discretization of $x(\cdot)$. Every breakpoint $z_i$ (save $z_0$ and $z_4$) is duplicated; the left copy of $z_i$ is allocated according to $x(z_{i-1})$ and the right copy is allocated according to $x(z_i)$ (the copies are plotted as if separated by a gap but the distance between $(L,i)$ and $(R,i)$ is vanishing.}
    \label{fig:Fig2}
\end{figure}{}







\paragraph{\bf Proof of Theorem~\ref{thm:char-cont}.}

Our proof of the characterization result for the continuous case relies on the following proposition (a slight generalization of a classic result by \cite{Myerson81}) 
and on
Lemma~\ref{lemma:reduction} below.

\begin{proposition}\label{prop:impl-is-mon-cont}
For continuous types, every implementable allocation rule is monotone piecewise constant.
\end{proposition}

\begin{proof}
We apply Myerson's argument to show that every implementable allocation rule $x$ mapping types to actions is monotone. Thus and since there is a finite number of actions to map to, $x$ is piecewise constant. 
To show that $x$ is monotone, 
let $c<c'$ be two costs; it suffices to show that $\gamma_{x{(c)}}\geq \gamma_{x{(c')}}$. Let $t$ be a payment rule that implements $x$, and recall $T^c_i$ is the expected payment to the agent for reporting type $c$ and taking action $i$. If the agent's true type is $c$, by IC it must hold that the expected payment from truthfulness is at least the expected payment from reporting $c'$ and deviating to the action recommended to type $c'$: 
\begin{equation}\label{eq:proof-menumonotonicity-eq1}
T^{c}_{x(c)}-\gamma_{x(c)} c\geq T^{c'}_{x(c')}-\gamma_{x(c')} c.
\end{equation}
Similarly, if the agent's true type is $c'$ then by IC it must hold that
\begin{equation}\label{eq:proof-menumonotonicity-eq2}
T^{c'}_{x(c')}-\gamma_{x(c')} c'\geq T^{c}_{x(c)}-\gamma_{x(c)} c'.
\end{equation}
Rearranging \eqref{eq:proof-menumonotonicity-eq1} and \eqref{eq:proof-menumonotonicity-eq2} we have
\begin{equation}\label{eq:proof-menumonotonicity-eq3}
(\gamma_{x(c)} -\gamma_{x(c')}) c' \geq  T^{c}_{x(c)}-T^{c'}_{x(c')} \geq (\gamma_{x(c)} -\gamma_{x(c')}) c.
\end{equation}
Suppose towards a contradiction that $\gamma_{x(c)} <\gamma_{x(c')}$. Dividing both sides of \eqref{eq:proof-menumonotonicity-eq3} by $\gamma_{x(c)} -\gamma_{x(c')}$ we get that $c'\leq c$, contradiction. This completes the proof.
\end{proof}

The following lemma reduces the implementability of a monotone piecewise constant allocation rule to that of the discretized version with an extra condition.

\begin{lemma}\label{lemma:reduction}
Consider a principal-agent instance $\Con(F,\gamma,C)$ and a monotone piecewise constant allocation rule $x:C\to[n]$ with a set $Z$ of breakpoints. 
The following statements are equivalent for every payment rule $t: C \rightarrow \mathbb{R}^{m+1}_{\geq 0}$:
\begin{enumerate}
    \item Payment rule $t$ implements $x$ with respect to $\Con(F,\gamma,C)$;
    \item Payment rule $t$ implements $x$ with respect to $\Dis(F,\gamma,Z)$, such that
    \begin{equation}
    \textstyle
    \forall 1\le i\le \ell : T^{z_i}_{x(z_i)} - \gamma_{x(z_i)} z_i = T^{z_{i-1}}_{x(z_{i-1})} - \gamma_{x(z_{i-1})} z_i.
    \label{item:disc-constraint}
    \end{equation}
\end{enumerate}
\end{lemma}

By Condition~\eqref{item:disc-constraint} in Lemma~\ref{lemma:reduction}, type $z_i$ is equally incentivized to take action $x(z_{i})$, which is allocated to types in the interval $[z_i,z_{i+1})$ on its right, and to deviate to $z_{i-1}$ and take action $x(z_{i-1})$, which is allocated to types in the interval $[z_{i-1},z_{i})$ on its left. Equivalently, in the discretization described above, type $(R,i)$ is incentivized to take its prescribed action $x(z_{i})$, and type $(L,i)$ (in the limit as its cost $c \rightarrow z_i$) is incentivized to take its prescribed action $x(z_{i-1})$. 

\begin{proof}[Proof of Lemma~\ref{lemma:reduction}]
For the forward direction, suppose there exists a payment rule $t$ that implements $x$ with respect to $\Con(F,\gamma,C)$. We show that $t$ restricted to $Z$ both (i)~implements $x$ with respect to $\Dis(F,\gamma,Z)$; and (ii)~satisfies Condition~\eqref{item:disc-constraint}. 
For (i), since $t$ implements $x$ with respect to $\Con(F,\gamma,C)$ and $Z\subseteq C$, then type $z_i$ cannot increase its expected utility by reporting $z_{i'}$ or by taking an action other than $x(z_i)$. For (ii), suppose towards a contradiction that for some $1\le i\le n$ it holds that $\textstyle T^{z_i}_{x(z_i)}-\gamma_{x(z_i)}z_i\neq T^{z_{i-1}}_{x(z_{i-1})}-\gamma_{x(z_{i-1})}z_i$. That is, without loss of generality, $\textstyle (\gamma_{x(z_{i-1})}-\gamma_{x(z_i)}) z_i< T^{z_{i-1}}_{x(z_{i-1})}- T^{z_i}_{x(z_i)}$. It follows that for a small enough value of $\epsilon$, $\textstyle (\gamma_{x(z_{i-1})}-\gamma_{x(z_i)}) (z_i+\epsilon)< T^{z_{i-1}}_{x(z_{i-1})}- T^{z_i}_{x(z_i)}$. This contradicts the assumption that $t$ implements $x$ with respect to $\Con(F,\gamma,C)$, since type $z_i+\epsilon\in C$ could increase its payoff by reporting $c'={z_{i-1}}$ and by taking action $x(z_{i-1})$.

For the backward direction, suppose $t$ implements $x$ with respect to $\Dis(F,\gamma,C)$ under Condition~\eqref{item:disc-constraint}. Define $\hat{t}$ such that $\hat{t}^c={t}^{z_{i}}$ for every $c\in [z_{i},z_{i+1})$, where $i\in [\ell]$, and $t^{z_{\ell+1}}=t^{z_{\ell}}$.\footnote{Note that $T^{z_{\ell}}_{x(z_{\ell})}=T^{z_{\ell+1}}_{x(z_{\ell+1})}$ since $x(z_{\ell})=x(z_{\ell+1})$ and none of types $z_{\ell},z_{\ell+1}$ deviates to the other.} Since $t^{z_{i}}$ and action $x(z_i)$ maximize type $z_{i}$'s expected utility among all types $Z$ and actions,
\begin{eqnarray}\label{eq:dis-to-cont-1}
\textstyle (\gamma_{k}-\gamma_{x(z_{i})}) z_{i} \geq T^{z_{i'}}_{k}-T^{z_i}_{x(z_i)} & \forall i\in [\ell],k\in [n], z_{i'} \in Z.
\end{eqnarray}
By Condition~\eqref{item:disc-constraint} and the fact that $t^{z_{i+1}}$ and action $x(z_{i+1})$ maximize type $z_{i+1}$'s expected utility $\forall i \in [\ell]$, we have that $t^{z_i}$ and action $x(z_i)$ also maximize type $z_{i+1}$'s expected utility. Thus, 
\begin{eqnarray}\label{eq:dis-to-cont-2}
\textstyle T^{z_i}_{x(z_i)} - T^{z_{i'}}_{k} \geq  (\gamma_{x(z_i)}-\gamma_{k}) z_{i+1} & \forall i\in [\ell],k\in [n], z_{i'}\in Z.
\end{eqnarray}
The above holds also for $i=\ell$ since ${x(z_{\ell})}={x(z_{\ell+1})}$, and we chose $t^{z_{\ell+1}}=t^{z_{\ell}}$ which also incentivizes $z_{\ell+1}$ to take $x(z_{\ell+1})$ as explained above.

Using Inequalities \eqref{eq:dis-to-cont-1}-\eqref{eq:dis-to-cont-2} we now show that type $c$'s expected utility is maximized by choosing payments $\hat{t}^{c}$ and action $x(c)$ for every $c\in [z_i,z_{i+1}), i\in [\ell]$. First, if $\gamma_{x(z_i)}>\gamma_{k}$, then by \eqref{eq:dis-to-cont-2} and using $\hat{t}^{c}={t}^{z_i}$, $x(c)=x(z_i)$, and $z_{i+1}> c$, we have that $\hat{T}^{c}_{x(c)} -\hat{T}^{c'}_{k}  \geq  (\gamma_{x(c)}-\gamma_{k}) c$  for every $k\in [n], c'\in C$. Otherwise, if $\gamma_{x(z_i)}\leq \gamma_{k}$, then by~\eqref{eq:dis-to-cont-1} and using that $\hat{t}^c=t^{z_i}$, $x(c)=x(z_i)$, and $c\geq z_{i}$, we have that $(\gamma_{k}-\gamma_{x(c)}) c\geq \hat{T}^{c'}_k-\hat{T}^c_{x(c)}$ for every $k\in [n], c'\in C$. This completes the proof. 
\end{proof}

We are now ready to prove our characterization result for continuous types.

\begin{proof}[Proof of Theorem~\ref{thm:char-cont}]
By combining Proposition~\ref{prop:impl-is-mon-cont} and Lemma~\ref{lemma:reduction}, to determine whether $x$ is implementable one must determine whether there exists a payment rule $t$ that implements $x$ for $\Dis(F,\gamma,Z)$ under Condition~(\ref{item:disc-constraint}). 
As in the proof of Theorem~\ref{thm:charac-disc} we adopt the following LP-based approach. The LP for finding an IC contract that implements $x$ under Condition~(\ref{item:disc-constraint}) (if one exists) has $|Z| m=O(nm)$ payment variables $\{t^{z_i}_j\}$, which must be non-negative by limited liability, and two sets of $O(|Z|^2n)=O(n^3)$ constraints referred to as the ``right'' set and the ``left'' set. 
The right set of constraints ensures that reporting $z_i$ and taking action $x(z_i)$ maximizes type $z_i$'s expected utility among all type reports and actions. This set of constraints is equivalent to the constraints in \eqref{LP1}, ensuring that $t$ implements $x$.
The left set of constraints ensures that reporting $z_{i-1}$ and taking action $x(z_{i-1})$ also maximizes type $z_i$'s expected utility (i.e., ensuring Condition~(\ref{item:disc-constraint})).%
\footnote{It is also possible to replace this set of constraints by an equality constraint, which leads to a ``less symmetric'' dual.}
\begin{equation}\label{LP:cont}\tag{LP2} 
\begin{array}{ll@{}ll}
\text{min} & 0 &\\
\text{s.t.} & T^{z_i}_{x(z_i)} -\gamma_{x(z_i)} z_i \geq
T^{z_{i'}}_k-\gamma_{k} z_i &&  \forall i,i' \in [\ell+1],k \in [n],\\
&T_{x(z_{i-1})}^{z_{i-1}} - \gamma_{x(z_{i-1})} z_i \geq T^{z_{i'}}_k- \gamma_{k} z_i &&  \forall 1\le i\le \ell ,i'\in [\ell+1],k \in [n],\\
&t^{z_i}_j\geq 0 &&  \forall i\in [\ell+1], j\in [m].
\end{array}
\end{equation}

The dual of \eqref{LP:cont} has two sets of $O(n^2 m)$ non-negative variables: ``right'' variables $\lambda_{(R,i,i',k)}$, and ``left'' variables $\lambda_{(L,i+1,i',k)}$ for for $i\in [\ell],i'\in [\ell+1],k\in [n]$. 
In our interpretation, the indexing  $(\{R,L\},i,i',k)$ corresponds to duplicate side, true type, reported type, and action. The dual is:

\begin{equation}\label{Dual2}\tag{D2} 
\begin{array}{ll@{}ll}
\text{max} & \sum_{i\in [\ell],i'\in [\ell+1],k\in [n]}\lambda_{(R,i,i',k)}(\gamma_{x(z_i)} z_i-\gamma_{k} z_i)+\\
&\sum_{i\in [\ell],i'\in [\ell+1],k \in [n]}\lambda_{(L,i+1,i',k)} (\gamma_{x(z_{i})} z_{i+1} - \gamma_{k} z_{i+1})\\
\text{s.t.} & \sum_{i'\in [\ell],k\in [n]}\lambda
_{(R,i',i,k)}F_{k,j}+ \sum_{i'\in [\ell],k\in [n]}\lambda
_{(L,i'+1,i,k)}F_{k,j}\geq &\\
&F_{x(z_i),j}(\sum_{i'\in [\ell+1],k\in [n]}\lambda
_{(R,i,i',k)}+\sum_{i'\in [\ell+1],k\in [n]}\lambda
_{(L,i+1,i',k)}) && \forall  j \in [m], i\in [\ell]\\
\end{array}
\end{equation}

By strong duality, \eqref{LP:cont} is feasible if and only if there are no variable assignments $\lambda$ such that the objective of \eqref{Dual2} is strictly positive and the constraints in \eqref{Dual2} satisfy. Lemma~\ref{lemma:normalize-cont-char} in  Appendix~\ref{appx:Char} shows that the existence of $\lambda$ implies the existence of a feasible solution $\lambda'$ with strictly positive objective in which $\sum_{i',k\in [n]}\lambda'_{(R,i,i',k)}=\sum_{i',k\in [n]}\lambda'_{(L,i+1,i',k)}=1$ $\forall i \in [\ell]$.

It can be verified that when the above holds, strictly positive objective value simplifies to $\sum_{i\in [\ell]}(\gamma_{x(z_i)} z_i+\gamma_{x(z_{i})} z_{i+1}) >\sum_{i\in [\ell],i'\in [\ell+1],k\in [n]}\lambda_{(R,i,i',k)}\gamma_{k} z_i+\sum_{i\in [\ell],i'\in [\ell+1],k \in [n]}\lambda_{(L,i+1,i',k)}\gamma_{k} z_{i+1}$  and the constraints simplify to $\sum_{i'\in [\ell],k\in [n]}\lambda_{(R,i',i,k)}F_{k,j}+ \sum_{i'\in [\ell],k\in [n]}\lambda
_{(L,i'+1,i,k)}F_{k,j}\geq 
2F_{x(z_i),j}$ $\forall  j \in [m], i\in [\ell]$. 
\end{proof}

 

\paragraph{\bf Algorithmic aspects.}

The LP-based proof of the characterization for the continuous setting again leads to efficient algorithms for checking implementability and finding optimal payments for an implementable allocation rule.
 
\begin{corollary}\label{cor:poly-LP-cont}
For continuous types, the problem of determining whether or not a piecewise constant allocation rule is implementable is solvable in time $O(poly(n,m))$. Moreover, if an allocation rule is implementable, then we can find optimal payments for this allocation rule in time $O(poly(n,m))$.
\end{corollary}

\section{Applying the Characterization}\label{sec:applications}

In this section we present two applications of the characterization to optimal contracts. First, we obtain a polynomial-time algorithm for computing the optimal contract when the number of actions is constant. Second, we give a non-trivial example of when a ``divide-and-conquer'' approach to computing the optimal contract can succeed: We show that separate treatment of the two types of IC constraints --- those that address private types and those that take care of hidden action --- works for the case of uniformly distributed costs. The idea is to initially ignore the IC constraints introduced by hidden action; this brings us into Myerson territory and we can optimize over monotone allocation rules. In a second step we then verify that for uniform distributions, this optimal monotone rule happens to satisfy the IC constraints due to hidden action.


\subsection{A Poly-Time Algorithm for a Fixed Number of Actions}
\label{sub:application1}

We give a polynomial-time algorithm for computing the optimal contract given discrete types $C$ and a constant number of actions $n$. 
%
This positive result is in sharp contrast to the multi-parameter model, where computing the best contract for a constant number of actions is APX-hard \cite[e.g.,][]{guruganesh2020contracts,CastiglioniM021}. 
In our single-parameter model, Proposition~\ref{prop:impl-is-mon-dis} allows us to search over \emph{monotone} allocation rules, reducing the complexity from $n^{|C|}$ (all allocations of actions to the $|C|$ types) to $|C|^n$.
We evaluate the rules by finding the optimal payments for each rule via Corollary~\ref{cor:poly-LP}.

\begin{theorem}
\label{thm:const-n}
The problem of computing the optimal contract for discrete types is solvable in polynomial time for a constant number of actions $n$.
\end{theorem}


\begin{proof}
According to Proposition~\ref{prop:impl-is-mon-dis}, to find the optimal implementable allocation rule it suffices to optimize over monotone rules. We bound the number of monotone allocation rules by the number of combinations of $|C|$ actions from $[n]$ (with repetition), by showing the each allocation rule uniquely maps to such a set. Let $x$ be a monotone allocation rule, take $S_x$ as the multiset of actions allocated in $x$. {Any different allocation rule $x'$ induces $S_{x}\neq S_{x'}$. To see this, let $c$ be the smallest cost for which $x(c)\neq  x'(c)$, assuming $x(c)> x'(c)$ without loss of generality. Observe that $x(c)$ will not be allocated for higher costs than $c$ in $x'$ by monotonicity of the allocation $x'$. This implies that the number of $x(c)$ instances in $S_{x}$ is strictly higher than that of $S_{x'}$, so $S_x\neq S_{x'}$.} 
It is known that the number of combinations with repetition of $|C|$ not-necessarily distinct elements from a set of size $n+1$ is $\binom{|C|+n}{|C|}$ (see e.g.~\citep{combi-wiki}). By definition, and then by reorganizing,
$$
\textstyle \binom{|C|+n}{|C|}=\frac{(|C|+n)!}{|C|!\cdot n!}=\frac{|C|!}{|C|!}\cdot\frac{ \prod_{i=1}^{n}(|C|+i)}{ n!}=O(|C|^n),
$$
for constant $n$. Thus, the number of monotone allocation rules in polynomial in $|C|$. Recall Corollary~\ref{cor:poly-LP}, stating that it is possible to determine whether or not an allocation rule is implementable, and to give optimal contract if so, in polynomial time. Thus, using the brute-force approach of computing the expected revenue (reward minus payments of an optimal contract) of all monotone allocation rules yields a polynomial time procedure for computing the optimal allocation. This completes the proof.
\end{proof}

\subsection{Optimal Contract for Uniformly-Distributed Costs}
\label{sub:application2}


In this section we explore a Myersonian approach to optimal contract design for continuous-type settings: optimizing the expected revenue over all monotone allocation rules, in hopes that for the \emph{revenue-optimal} such rule, monotonicity turns out to be sufficient for implementability. 
In Theorem~\ref{thm:uniform} we show this to be the case for uniformly distributed costs.

\begin{assumption}\label{assumption:high-cost-type}
Consider the highest-cost type $\bar{c}$. For every action $ 1\leq i\le n$, the expected reward is strictly less than this type's cost, i.e.,  $\gamma_i \bar{c}>R_i$. 
\end{assumption}

Intuitively, the above assumption implies that when facing type~$\bar{c}$, the only action the principal can incentivize without losing herself is the zero-cost action.
This assumption will come in handy in Section~\ref{sub:opt-cont-uniform}, where we further discuss it.

\subsubsection{Generalization of the Myerson Toolbox to Hidden Action}


We start by generalizing several classic results by \citet{Myerson81} to our hidden action model with continuous types. Like Myerson's theory, the results of this section apply to \emph{randomized} allocation rules, defined as follows:

\begin{definition}\label{def:rand-alloc}
    A \emph{randomized} allocation rule $x:C\to\Delta([n])$ is a mapping from a cost $c$ to a distribution over recommended actions, where $x_k(c)$ denotes the probability 
    of recommending action $k\in [n]$.
\end{definition}

Overloading notation to accommodate for randomization, denote by $R_{x(c)}=\sum_{k\in [n]}x_k(c) R_{k}$ the expected reward of $x$ given type $c$, and by $\gamma_{x(c)}=\sum_{k\in [n]}x_k(c) \gamma_k$ the expected effort. As before, the payment rule maps from cost $c$ to a vector of payments, which can now be seen as expected payments over the random coin tosses. Formally, let $T_{x(c)}^c = \sum_{i \in [n]}\sum_{j \in [m]} x_i(c) \cdot  F_{i,j} \cdot t^c_j$ denote the expected payment over both random actions and random outcomes.

\begin{lemma}[Payment identity]
\label{lemma:myerson}
Let $C=[0,\bar{c}]$ be a continuous type space. Let  $x:C\to\Delta([n])$ be a randomized allocation rule. Denote by  $\gamma'_{x(c)}$ the derivative of $\gamma_{x(c)}$ with respect to $c$. Then if $x$ is implementable, for every payment rule $t$ that implements $x$,
\begin{eqnarray*}
\textstyle T^c_{x(c)}=T^{\bar{c}}_{x(\bar{c})}-\int_{c}^{\bar{c}} z\gamma'_{x(z)}  \,dz & \forall c\geq 0.
\end{eqnarray*}
\end{lemma}

\begin{proof} Let $x$ be an implementable allocation rule. From similar arguments as in the proof of Proposition~\ref{prop:impl-is-mon-cont}, we have that for every contract $t$ that implements $x$, 
\begin{eqnarray*}\label{eq:proof-menumonotonicity-eq4}
\textstyle(\gamma_{{x}(c)} -\gamma_{{x}(c')}) c' \geq  T^{c}_{x(c)}-T^{c'}_{x(c')}\geq (\gamma_{{x}(c)} -\gamma_{{x}(c')}) c & \forall c<c'\in C.
\end{eqnarray*}
Dividing the above by $c'-c$ we get
\begin{eqnarray*}
\textstyle-\frac{[\gamma_{{x}(c')} -\gamma_{{x}(c)}]}{c'-c}\cdot c' \geq  -\frac{T^{c'}_{x(c')}-T^{c}_{x(c)}}{c'-c}\geq -\frac{[\gamma_{{x}(c')} -\gamma_{{x}(c)}]}{c'-c}\cdot c & \forall c<c'.
\end{eqnarray*}
Taking the limit as $c' \downarrow c$ yields the following constraint: $T'^c_{x(c)} =\gamma'_{x(c)}\cdot c$ $\forall c\in C.$
Thus, $\int_{c}^{\bar{c}} T'^c_{x(c)} \,dc =\int_{c}^{\bar{c}} z\gamma'_{x(z)}\,dz$ $\forall c\in C.$ 
Thus, $T^{c}_{x(c)}\mid_{c}^{\bar{c}} =T^{\bar{c}}_{x(\bar{c})}-T^{c}_{x(c)}=\int_{c}^{\bar{c}} z\gamma'_{x(z)}\,dz$ $\forall c\geq 0.$ This implies that $T^{c}_{x(c)}=T^{\bar{c}}_{x(\bar{c})}-\int_{c}^{\bar{c}} z\gamma'_{x(z)} \,dz$ $\forall c\geq 0.$
\end{proof}

An appropriate variant of the classic result that expected revenue equals expected virtual welfare applies in our model.
We first define for completeness the notion of virtual costs (the ``reverse'' version of virtual values by \cite{Myerson81}): 

\begin{definition}
 Let $G$ be a distribution over a continuous type set $C$ with density $g$. Given cost $c\in C$, the \emph{virtual cost} is $$\varphi(c)=c+\frac{G(c)}{g(c)}.$$
\end{definition}

We can then define the expected revenue of a contract as $\mathbb{E}_{c \sim G}[R_{x(c)}-T^c_{x(c)}]$, and its expected virtual welfare as $\mathbb{E}_{c \sim G}[R_{x(c)} - \varphi(c) \gamma_{x(c)}]$.

\begin{proposition}
\label{prop:rev-wel} 
For continuous types, the expected revenue of an IC contract $(x,t)$ is equal to its expected virtual welfare minus the expected utility of type $\bar{c}$:
$$
\mathbb{E}_{c\sim G}[R_{x{(c)}}-T^c_{x(c)}]=\mathbb{E}_{c\sim G}[R_{x{(c)}}-\varphi(c) \gamma_{x{(c)}}]-(T^{\bar{c}}_{x(\bar{c})}-\gamma_{x(\bar{c})}\bar{c}).
$$
\end{proposition}

\begin{proof}
We show that the expected payment is equal to the expected virtual cost, i.e., $\mathbb{E}_{c\sim G}[T^c_{x(c)}]$ $=T^{\bar{c}}_{x(\bar{c})}-\gamma_{x(\bar{c})}\bar{c}+\mathbb{E}_{c\sim G}[\varphi(c) \gamma_{x(c)}]$. By linearity of expectation this suffices to prove the proposition. By the expected payment formula in Lemma~\ref{lemma:myerson},
\begin{eqnarray*}
\textstyle\mathbb{E}_{c\sim G}[T^c_{x(c)}]=\int_{0}^{\bar{c}}T^c_{x(c)} g(c) \dd c =T^{\bar{c}}_{x(\bar{c})}G(\bar{c}) + \int_{0}^{\bar{c}}[-\int_{c}^{\bar{c}} z\gamma'_{x(z)}  \dd z] g(c) \dd c. 
\end{eqnarray*}
Reversing the integration order,
\begin{eqnarray*}
{T^{\bar{c}}_{x(\bar{c})}G(\bar{c}) + \textstyle\int_{0}^{\bar{c}}-[\int_{0}^{z}g(c)  \dd c] z\gamma'_{x(z)} \dd z =
T^{\bar{c}}_{x(\bar{c})}G(\bar{c}) + \int_{0}^{\bar{c}}-G(z) \gamma'_{x(z)} z \dd z.}
\end{eqnarray*}
Using integration by parts, where $G(z)z=u(z)$ and $\gamma'_{x(z)}=v'(z)$
\begin{eqnarray*}
T^{\bar{c}}_{x(\bar{c})}G(\bar{c}) - \textstyle\int_{0}^{\bar{c}}{G(z) z}\cdot  {\gamma'_{x(z)}} \dd z = T^{\bar{c}}_{x(\bar{c})}G(\bar{c}) -G(z) z  \gamma_{x(z)}\mid^{\bar{c}}_{0}+\int_{0}^{\bar{c}}(g(z) z+G(z)) \gamma_{x(z)} \dd z.
\end{eqnarray*}
Since $G(\bar{c}) = 1$, $G(0) = 0$, the above equals to $T^{\bar{c}}_{x(\bar{c})}-\gamma_{x(\bar{c})}\bar{c}+\int_{0}^{\bar{c}}( z+\frac{G(z)}{g(z)}) \gamma_{x(z)} g(z) \dd z$. That is, 
\begin{eqnarray*}
T^{\bar{c}}_{x(\bar{c})}-\gamma_{x(\bar{c})}\bar{c}+\mathbb{E}_{c\sim G}[\varphi(c)\gamma_{x(c)}].
\end{eqnarray*}
This completes the proof.
\end{proof}

\subsubsection{Finding the Optimal Contract for Uniform Costs}\label{sub:opt-cont-uniform}

According to Proposition \ref{prop:rev-wel}, the expected revenue of an IC contract is equal to the expected virtual welfare minus the (non-negative) expected utility of the highest type. 
Focusing on the first term and ignoring momentarily the second one, a natural candidate for the revenue-maximizing allocation rule is the one that maximizes virtual welfare.

\begin{definition}\label{def:welfare-max-alloc}
The \emph{virtual welfare-maximizing allocation rule} $x^*$ (among all randomized such rules) is given by the deterministic rule that chooses
\begin{eqnarray*}
x^*(c)=\arg\max_{i\in [n]}\{R_i-\varphi(c)\cdot \gamma_i\} && \forall c\in C.
\end{eqnarray*}
\end{definition}





The following theorem shows that when the distribution over types is uniform, the virtual welfare maximizing allocation rule is implementable --- a non-trivial result in light of Proposition~\ref{prop:non-monotone}, which shows that monotonicity alone is insufficient. 

Furthermore, using Assumption~\ref{assumption:high-cost-type} this theorem shows that $x^*$ is implementable by a contract for which type $\bar{c}$'s expected utility is exactly $0$ (this is the reason we could ignore the second term above). Then,  Proposition~\ref{prop:rev-wel} implies optimality of $x^*$.


\begin{theorem}
\label{thm:uniform}
Assuming~\ref{assumption:high-cost-type}, and uniform distribution over types $G=U[0,\bar{c}]$, the virtual welfare maximizing allocation rule $x^*$ is implementable by a payment rule $t$ for which $T^{\bar{c}}_{x^*(\bar{c})}-\gamma_{x^*(\bar{c})}\bar{c}$=0.
Contract $(x^*,t)$ thus maximizes expected revenue among all IC, limited liability contracts.
\end{theorem}

\begin{proof}

Let $x^*$ be the virtual-welfare maximizing allocation rule. Note that by the assumption that $R_i< \gamma_i\bar{c}$ for $1\le i\le n$ and since $\bar{c} \le \varphi(\bar{c})$, we have that $x^*(\bar{c})=0.$ Thus, to show that $x^*$ is implementable by a contract $t$ for which $T^{\bar{c}}_{x^*(\bar{c})}-\gamma_{x^*(\bar{c})}\bar{c}=0$, it suffices to show that $x^*$ is implementable by a contract in which $t^c_0=0$ $\forall c\in C$. To test whether there exists such a contract, we test whether~\eqref{LP:cont} is feasible for $x^*$ when restricting $t^{z_i}_0=0$ $\forall i\in [\ell+1]$. Specifically, we test the feasibility of the following linear program.
\begin{equation*}
\begin{array}{ll@{}ll}
\min & 0 &\\
\text{s.t.} & \sum_{j=1}^m F_{x^*(z_i),j}t^{z_i}_j-\gamma_{x^*(z_i)} z_i \geq
\sum_{j=1}^m F_{k,j}t^{z_{i'}}_j-\gamma_{k} z_i &&  \forall i,i' \in [\ell+1],k \in [n],\\
& \sum_{j=1}^m F_{x^*(z_{i-1}),j}t^{z_{i-1}}_j - \gamma_{x^*(z_{i-1})} z_i \geq \sum_{j=1}^m F_{k,j}t^{z_{i'}}_j - \gamma_{k} z_i &&  \forall 1\le i\le \ell ,i'\in [\ell+1],k \in [n],\\
&t^{z_i}_j\geq 0 &&  \forall i\in [\ell+1], 1\le j\le m.
\end{array}
\end{equation*}
Using the same techniques as in the proof for Theorem~\ref{thm:char-cont} we have that $x^*$ is implementable by a contract for which $t^{z_i}_0=0$ $\forall i\in [\ell+1]$ if and only if there exist no weights $\lambda_{(L,i+1,{i'},k)},$ $\lambda_{(R,i,{i'},k)} \geq 0$ for all $i\in [\ell],i'\in[\ell+1],k \in [n]$ which satisfy $\sum_{i'\in [\ell+1],k\in [n]}\lambda_{(L,i+1,i',k)}= \sum_{i'\in [\ell+1],k\in [n]}\lambda_{(R,i,i',k)}=1$ $\forall i\in [\ell]$ and the following conditions: (1) $\sum_{i'\in [\ell],k\in [n]}\frac{1}{2}(\lambda_{(R,i',i,k)}F_{k,j}+ \lambda
_{(L,i'+1,i,k)}F_{k,j})\geq 
F_{x^*(z_i),j}$ $\forall  1\le j \le m, i\in [\ell],$ and (2) $\sum_{i\in [\ell]}(\gamma_{x^*(z_i)} z_i+\gamma_{x^*(z_{i})} z_{i+1}) >\sum_{i\in [\ell],i'\in [\ell+1],k\in [n]}\lambda_{(R,i,i',k)}\gamma_{k} z_i+\lambda_{(L,i+1,i',k)}\gamma_{k} z_{i+1}.$

Suppose towards a contradiction that $x^*$ is unimplementable by such a contract. Fix a sufficiently small value of $\epsilon$ to be determined and define the following randomized allocation rule.
\begin{eqnarray}\label{eq:new-alloc}
x_k(c)=\begin{cases}  
\sum_{i'\in [\ell+1]}\lambda_{(R,i,i',k)} & c\in [z_i,z_i+\epsilon),i\in [\ell],\\
\sum_{i'\in [\ell+1]}\lambda_{(L,i+1,i',k)} & c\in [z_{i+1}-\epsilon,z_{i+1}),i\in [\ell],\\
x_k^*(c) & \text{otherwise,}
\end{cases}  & \forall k\in [n].
\end{eqnarray}
In Appendix~\ref{appx:uni} we show that for sufficiently small value of $\epsilon$, the expected virtual welfare of $x$ is strictly higher than that of $x^*$, contradicting our hypothesis that $x^*$ is the virtual welfare maximizer.
\end{proof}

A challenge in generalizing the result in Theorem~\ref{thm:uniform} to other distributions is that in our proof we used that for uniform distributions, strictly lower joint cost implies strictly lower joint \emph{virtual} cost. This is not the case for other distributions, e.g., exponential as the following example shows. 


\begin{example} Consider a setting with $5$ actions with $\gamma_0=0$, $\gamma_1=1,\gamma_2=2,\gamma_3=3,\gamma_4=7$ and two types $c_1=1$, $c_1=2$. The joint cost of $x(1)=2$, $x(2)=3$ is: $\gamma_3\cdot 2+\gamma_2\cdot 1=8$. The joint cost of $x'(1)=4, x'(2)=1$ is: $\gamma_1\cdot 2+\gamma_4\cdot 1=9$. Thus the joint cost of $x$ is lower than the joint cost of $x'$. The virtual cost function of exponential distribution with parameter $1$ is $\varphi(c)=c+e^c-1$. Thus, The joint virtual cost of $x(1)=2$, $x(2)=3$ is: $\gamma_3\cdot (2+e^2-1)+\gamma_2\cdot (1+e-1)\approx 30.6$. The joint cost of $x'(1)=4, x'(2)=1$ is: $\gamma_1\cdot (2+e^2-1)+\gamma_4\cdot (1+e-1)\approx 27.41$. Thus, the joint virtual cost of $x$ is higher than the joint virtual cost of $x'$.
\end{example}

In Appendix~\ref{appx:beyond-uni} we discuss additional computational results and structural insights for other natural classes of distributions.

\section{Conclusion}

In this work we propose a natural principal-agent problem with hidden action and single-dimensional private types, which seems to strike a balance between applicability and tractability. We provide analogs to Myerson's result --- two characterizations, one for discrete and one for continuous types --- which we believe the agenda of CT $\times$ MD can build upon.

We provide two proofs of concept to show this direction is promising: (1)~a ``more positive'' computational result than the corresponding result in the multi-parameter model of \citet{guruganesh2020contracts}; and (2)~a non-trivial ``divide and conquer'' result, which suggests that the optimal monotone rule could actually be implementable (at least somewhat) generally, despite the fact that monotonicity is not a sufficient condition in our model.

An interesting question going forward is to verify or disprove that this approach works for all regular distribution functions, or all regular distributions under some additional regularity assumption on the contract setting. A natural candidate for the latter is to assume diminishing marginal returns, which we show results in a ``nice'' allocation rule (Theorem~\ref{thm:regular-implies-monotone}), and ensures all actions can be incentivized absent hidden types.

\section*{Acknowledgements}

The authors would like to thank Yingkai Li for his helpful comments, which contributed to the paper in general and to Section~\ref{sub:application2} in particular.



\appendix
\bibliographystyle{ACM-Reference-Format}
\bibliography{bibliography}

\appendix
\numberwithin{equation}{section}

\section{Relation to Classic Literature}\label{appx:relation-lit}

For completeness, we show how Theorem~\ref{thm:charac-disc} reduces to the classic characterizations when, as specified in Section~\ref{sec:comparison}, there is either no hidden type (standard contract setting) or no hidden action (standard procurement setting). For the former, suppose $|C|=1$. Then $c=c'=c''$ and we can omit these from the notation. Then, the two conditions of Theorem~\ref{thm:charac-disc} reduce to the known condition for implementability of action $x$ under public costs $\gamma_1 c,\dots,\gamma_n c$ for the actions (see e.g. Proposition 2 in \citep{DuttingRT19}): There should be no distribution $\{\lambda_k\}_{k\in[n]}$ over the actions such that $\sum_{k\in [n]}  \lambda_{k}F_{k,j} \geq F_{x,j}$  for every $j \in [m]$ (that is, the weighted combination of the action distributions is equal to that of action $x$), and $\sum_{k\in [n]} \lambda_{k}\gamma_k c < \gamma_{x} c$ (that is, the combined cost is lower than that of action $x$).

As for the relation to no hidden action, consider a deterministic auction allocation rule $x$ of a procurement auction, which, for simplicity, determines from which types to buy. That is, $x(c)=1$ if type $c$ is allocated (bought from) and $x(c)=0$ otherwise. Notice that such $x$ can also be viewed as an allocation rule in a setting with two actions $0,1$, where $\gamma_0=0$, $\gamma_1=1$. An auction allocation rule is known to be implementable if and only if it is monotone in the auction sense, i.e., if for every two types $c'<c$ such that type $c$ wins (is allocated), type $c'$ also wins.

To see the connection to Theorem~\ref{thm:charac-disc}, consider a bipartite graph $((V,U),E)$, with $V$ and $U$ corresponding to two copies of the type set $C$. Let $W$ be the set of winning types according to auction allocation rule $x$.
It can be verified that the characterization in Theorem~\ref{thm:charac-disc} says $x$ is implementable if and only if there is no \emph{perfect} fractional matching $\{\lambda_{c,c'}\}$ in the bipartite graph for which
\begin{equation}
    \sum_{c\in W} c>\sum_{c\in C,c'\in W}\lambda_{c,c'}c.\label{eq:bipart}
\end{equation}
Such a matching corresponds to a fractional deviation plan where type $c$ reports $c'$ with probability $\lambda_{c,c'}$, 
and by perfection of the matching $\forall c : \sum_{c'} \lambda_{c',c}=1.$
We show that the existence of such a deviation plan is equivalent to non-monotonicity. If $x$ is not monotone, i.e., if there exist $c<c'$ such that $c\notin W$, and $c'\in W$, then a perfect fractional matching which satisfies condition~\eqref{eq:bipart} is when both $c$ and $c'$ swap. If there exists such a perfect matching, take the minimum $c\notin W$ such that $\lambda_{c,c'}>0$ for $c'\in W$.\footnote{Necessarily there exists such $c \notin W$. Otherwise, by the perfection of the matching, the right-hand side of~\eqref{eq:bipart} is $\sum_{c,c'\in W}\lambda_{c,c'}c=\sum_{c\in W}c$.} It holds that $c<\max_{c'\in W} c'$, otherwise, the right-hand side of~\eqref{eq:bipart} is at least $\sum_{c,c'\in W}\lambda_{c,c'}c+\sum_{c\notin W,c'\in W}\lambda_{c,c'} \max_{c'\in W} c'$, which is at least $\sum_{c\in W} c.$ We conclude that there exist $c\notin W$, and $c'\in W$ such that $c<c'$. That is, $x$ is not monotone.

\section{LP-Based Characterization}\label{appx:Char}

In this appendix we complete the lemmas used in the proofs of Theorem~\ref{thm:charac-disc} and of Theorem~\ref{thm:char-cont}.

\begin{lemma}\label{lemma:normalize-dis-char}
If there exists a feasible solution $\lambda$ to~\eqref{Dual1} with strictly negative objective value, then there exists a feasible solution $\lambda'$ with strictly negative objective in which $\sum_{c'\in C,k\in [n]}\lambda'_{(c,c',k)}=1$ $\forall c\in C$.
\end{lemma}

\begin{proof}
Let $M=\max_{c''\in C}{\sum_{c'\in C}\sum_{k\in [n]}\lambda_{(c'',c',k)}}$. Define $\lambda'_{(c,c', k)}$ $\forall c,c'\in C, k\in [n]$ as follows.
\begin{eqnarray*}
\lambda'_{(c,c',k)}=
\begin{cases}
1-\frac{1}{M}(\sum_{c\neq c'\in C}\sum_{k\in [n]}\lambda_{(c,c',k)}+\sum_{x(c)\neq k\in [n]}\lambda_{(c,c,k)}) & c'=c \wedge k=x(c)\\
\frac{1}{M}{\lambda_{(c,c',k)}} & \text{otherwise}.
\end{cases} 
\end{eqnarray*}
First, note that by definition of $\lambda'$, it holds that $\sum_{c'\in C}\sum_{k\in [n]}\lambda'_{(c,c',k)}=1$ $\forall c\in C$. Second, note that the second set of constraints is satisfied, i.e., $\lambda'_{(c,c',k)}\geq 0$ $\forall c,c'\in C,k\in [n]$. 
We next show that $\lambda'$ satisfies the first set of constraints. By reorganizing and splitting the summation on both sides, the constraints are satisfied if and only if
\begin{eqnarray*}
\textstyle F_{x(c),j} \lambda'_{(c,c,x(c))}+\sum_{c\neq c'\in C}\sum_{k\in [n]}F_{k,j} \lambda'_{(c,c',k)}+\sum_{x(c)\neq k\in [n]}F_{k,j} \lambda'_{(c,c,k)} &\geq& \\
\textstyle F_{x(c),j} \lambda'_{(c,c,x(c))}+\sum_{c\neq c''\in C}\sum_{k\in [n]}F_{x(c),j}\lambda'_{(c,c'',k)}+\sum_{x(c)\neq k\in [n]}F_{x(c),j} \lambda'_{(c,c,k)} && \forall c\in C, j\in [m].
\end{eqnarray*}
Subtracting $F_{x(c),j} (\lambda'_{(c,c,x(c))}-\frac{1}{M}{\lambda_{(c,c,x(c))}})$ from both sides, using the definition of $\lambda'$, and multiplying both sides by $M$, the above holds if and only if
\begin{eqnarray*}
\textstyle {F_{x(c),j} \lambda_{(c,c,x(c))}+\sum_{c\neq c'\in C}\sum_{k\in [n]}F_{k,j} \lambda_{(c,c',k)}+\sum_{x(c)\neq k\in [n]}F_{k,j} \lambda_{(c,c,k)}} &\geq& \\
\textstyle {F_{x(c),j} \lambda_{(c,c,x(c))}+\sum_{c\neq c''\in C}\sum_{k\in [n]}F_{x(c),j}\lambda_{(c,c'',k)}+\sum_{x(c)\neq k\in [n]} F_{x(c),j}\lambda_{(c,c,k)} } && \forall c\in C, j\in [m].
\end{eqnarray*}
Reorganizing the summation we have that the above holds if and only if
\begin{eqnarray*}
\textstyle \sum_{c'\in C}\sum_{k\in [n]}F_{k,j} \lambda_{(c,c',k)}\geq F_{x(c),j}\sum_{ c''\in C}\sum_{ k\in [n]} \lambda_{(c,c'',k)} && \forall c\in C, j\in [m].
\end{eqnarray*}
Since the above holds by the choice of $\lambda$, we have that $\lambda'$ satisfies the first set of constraints. We are left with showing that $\lambda'$ gives strictly negative objective value. 

Observe that when splitting the summation as done earlier and by the definition of $\lambda'$ we have that $\lambda'$ has strictly negative objective value if and only if 
\begin{eqnarray*}
\textstyle \sum_{c\in C}(1-\frac{1}{M}(\sum_{c\neq c'\in C}\sum_{k\in [n]}\lambda_{(c,c',k)}+\sum_{x(c)\neq k\in [n]}\lambda_{(c,c,k)})) (\gamma_{x(c)}-\gamma_{x(c)})c &+&\\
\textstyle \sum_{c\in C}\frac{1}{M}(\sum_{c\neq c'\in C}\sum_{k\in [n]} \lambda_{(c,c',k)} (\gamma_k-\gamma_{x(c)})+\sum_{x(c)\neq k\in [n]} \lambda_{(c,c,k)} (\gamma_k-\gamma_{x(c)}))c
&& < 0.
\end{eqnarray*}
By subtracting the first term $\sum_{c\in C}(1-\frac{1}{M}(\sum_{c\neq c'\in C}\sum_{k\in [n]}\lambda_{(c,c',k)}+\sum_{x(c)\neq k\in [n]}\lambda_{(c,c,k)}))(\gamma_{x(c)}-\gamma_{x(c)})c=0$ from both sides, adding $\sum_{c\in C}\frac{1}{M}\lambda_{(c,c,x(c))}(\gamma_{x(c)}-\gamma_{x(c)})c=0$ to both sides, and multiplying the inequality by $M$ we have that the above holds if and only if 
\begin{eqnarray*}
\textstyle \sum_{c,c'\in C}\sum_{k\in [n]} \lambda_{(c,c',k)} (\gamma_k-\gamma_{x(c)})c
< 0.
\end{eqnarray*}
Similarly, the above holds by the choice of $\lambda$.

From the above discussion we can conclude that \eqref{LP1} is feasible if and only if the dual has no feasible solution $\lambda'$ in which $\sum_{c'\in C}\sum_{k\in [n]}\lambda'_{(c,c',k)}=1$ $\forall c\in C$ with strictly negative objective.
\end{proof}

\begin{lemma}\label{lemma:normalize-cont-char}
If there exists a feasible solution $\lambda$ to~\eqref{Dual2} with strictly positive objective value, then there exists of a feasible solution $\lambda'$ with strictly positive objective in which $\sum_{i',k\in [n]}\lambda_{(L,i+1,i',k)}=\sum_{i',k\in [n]}\lambda_{(R,i,i',k)}=1$ $\forall i\in [\ell]$.
\end{lemma}

The proof uses the same normalization technique as presented above.

\begin{proof}
Choose $M$ such that $M\geq \max_{i\in [\ell]}{\sum_{i'\in [\ell+1],k\in [n]}\lambda_{(L,i+1,i',k)},\sum_{i' \in [\ell+1]\in,k\in [n]}\lambda_{(R,i,i',k)}}.$ Then, define $\lambda'$ as follows.
\begin{eqnarray*}
\lambda'_{(R,i,i',k)}=
\begin{cases}
1-\frac{1}{M}(\sum_{i\neq i'\in [\ell+1],k\in [n]}\lambda_{(R,i,i',k)}+\sum_{x(z_i)\neq k\in [n]}\lambda_{(R,i,i,k)}) & i'=i \wedge k=x(z_i),\\
\frac{1}{M}{\lambda_{(R,i,i',k)}} & \text{otherwise},
\end{cases}
\end{eqnarray*}

\begin{eqnarray*}
\lambda'_{(L,i+1,i',k)}=
\begin{cases}
1-\frac{1}{M}(\sum_{i\neq i'\in [n],k\in [n]}\lambda_{(L,i+1,i',k)}+\sum_{x(z_{i})\neq k\in [n]}\lambda_{(L,i+1,i,k)}) & i'=i \wedge k=x(z_{i}),\\
\frac{1}{M}{\lambda_{(L,i+1,i',k)}} & \text{otherwise},
\end{cases}
\end{eqnarray*}
for all $i\in [\ell], i'\in [\ell+1],k\in [n].$

First, note that by definition of $\lambda'$, it holds that $\lambda_{(R,i,i',k)},\lambda_{(L,i+1,i',k)} \geq 0$ $\forall i\in [\ell],i'\in [\ell+1],k\in [n]$, and that $\sum_{i'\in [\ell+1],k\in [n]}\lambda_{(R,i,i',k)}=\sum_{i'\in [\ell+1],k\in [n]}\lambda_{(L,i+1,i',k)}=1$ $\forall i\in [\ell]$. We next show that $\lambda'$ satisfies the set of constraints. By reorganizing and splitting the summation on both sides, the constraints are satisfied if and only if
\begin{eqnarray}\label{eq:proof1}
\nonumber \textstyle \sum_{x(z_i)\neq k\in [n]}F_{k,j}(\sum_{i'\in [\ell]}\lambda'_{(R,i',i,k)}  + \sum_{i'\in [\ell]}\lambda'_{(L,i'+1,i,k)})&+&\\
\nonumber \textstyle F_{x(z_i),j}(\sum_{i\neq i'\in [\ell]}\lambda'_{(R,i',i,x(z_i))}  + \sum_{i\neq i'\in [\ell]}\lambda'_{(L,i'+1,i,x(z_i))}) &+&\\
\textstyle F_{x(z_i),j} (\lambda'_{(R,i,i,x(z_i))}  + \lambda'_{(L,i+1,i,x(z_i))}) &\geq& \\
\nonumber\textstyle \sum_{x(z_i)\neq k\in[n]}F_{x(z_i),j} (\sum_{i'\in [\ell+1]}\lambda'_{(R,i,i',k)}+\sum_{i'\in[\ell+1]}\lambda'_{(L,i+1,i',k)}) &+&\\
\nonumber \textstyle F_{x(z_i),j} (\sum_{i\neq i'\in[\ell+1]}\lambda'_{(R,i,i',x(z_i))}+\sum_{i\neq i'\in[\ell+1]}\lambda'_{(L,i+1,i',x(z_i))})&+&\\
\nonumber \textstyle F_{x(z_i),j} (\lambda'_{(R,i,i,x(z_i))}+\lambda'_{(L,i+1,i,x(z_i))}) && \forall j\in [m],i\in [\ell].\nonumber 
\end{eqnarray}
By subtracting $F_{x(z_i),j} (\lambda'_{(R,i,i,x(z_i))}+\lambda'_{(L,i+1,i,x(z_i))})-\frac{1}{M}F_{x(z_i),j} (\lambda_{(R,i,i,x(z_i))}+\lambda_{(L,i+1,i,x(z_i))})$ from both sides in \eqref{eq:proof1}, and by definition of $\lambda'$ we have that the above holds if and only if 
\begin{eqnarray*}
\nonumber \textstyle \sum_{x(z_i)\neq k\in [n]}F_{k,j}(\sum_{i'\in [\ell]}\frac{1}{M}\lambda_{(R,i',i,k)}  + \sum_{i'\in [\ell]}\frac{1}{M}\lambda_{(L,i'+1,i,k)})&+&\\
\nonumber \textstyle F_{x(z_i),j}(\sum_{i\neq i'\in [\ell]}\frac{1}{M}\lambda_{(R,i',i,x(z_i))}  + \sum_{i\neq i'\in [\ell]}\frac{1}{M}\lambda_{(L,i'+1,i,x(z_i))}) &+&\\
\textstyle F_{x(z_i),j} (\frac{1}{M}\lambda_{(R,i,i,x(z_i))}  + \frac{1}{M}\lambda_{(L,i+1,i,x(z_i))}) &\geq& \\
\nonumber\textstyle \sum_{x(z_i)\neq k\in[n]}F_{x(z_i),j} (\sum_{i'\in [\ell+1]}\frac{1}{M}\lambda_{(R,i,i',k)}+\sum_{i'\in[\ell+1]}\frac{1}{M}\lambda_{(L,i+1,i',k)}) &+&\\
\nonumber \textstyle F_{x(z_i),j} (\sum_{i\neq i'\in[\ell+1]}\frac{1}{M}\lambda_{(R,i,i',x(z_i))}+\sum_{i\neq i'\in[\ell+1]}\frac{1}{M}\lambda_{(L,i+1,i',x(z_i))})&+&\\
\nonumber \textstyle F_{x(z_i),j} (\frac{1}{M}\lambda_{(R,i,i,x(z_i))}+\frac{1}{M}\lambda_{(L,i+1,i,x(z_i))}) && \forall j\in [m],i\in [\ell].\nonumber 
\end{eqnarray*}
By multiplying both inequalities by $M$, the above are exactly the constraints of \eqref{Dual2}, which hold for $\lambda$ by assumption.
Next, we show that the objective is strictly positive. Again, by reorganizing and splitting the summations, the objective becomes
\begin{eqnarray*}
\textstyle \sum_{i\in [\ell],i\neq i'\in [\ell+1],k\in [n]}\lambda'_{(R,i,i',k)}(\gamma_{x(z_i)} z_i-\gamma_{k} z_i)&+&\\
\textstyle \sum_{i\in [\ell],x(z_i)\neq k\in [n]}\lambda'_{(R,i,i,k)}(\gamma_{x(z_i)} z_i-\gamma_{k} z_i)&+&\\
\textstyle \sum_{i\in [\ell]}\lambda'_{(R,i,i,x(z_i))}(\gamma_{x(z_i)} z_i-\gamma_{x(z_i)} z_i)&+&\\
\textstyle \sum_{i\in [\ell],i\neq i'\in [\ell+1],k \in [n]}\lambda'_{(L,i+1,i',k)} (\gamma_{x(z_{i})} z_{i+1} - \gamma_{k} z_{i+1})\\
\textstyle \sum_{i\in [\ell],x(z_i)\neq k \in [n]}\lambda'_{(L,i+1,i,k)} (\gamma_{x(z_{i})} z_{i+1} - \gamma_{k} z_{i+1})\\
\textstyle \sum_{i\in [\ell]}\lambda'_{(L,i+1,i,x(z_i))} (\gamma_{x(z_{i})} z_{i+1} - \gamma_{x(z_i)} z_{i+1})\\
\end{eqnarray*}
After adding $\sum_{i\in [\ell]}\frac{1}{M}\lambda_{(R,i,i,x(z_i))}  (\gamma_{x(z_{i})}z_{i}-\gamma_{x(z_{i})}z_{i})$  $+\sum_{i\in [\ell]} \frac{1}{M}\lambda_{(L,i+1,i,x(z_i))} (\gamma_{x(z_{i})}z_{i+1}-\gamma_{x(z_{i})}z_{i+1})$ to the above, and using the definition of $\lambda$ we have that the above is strictly positive if and only if 
\begin{eqnarray*}
\textstyle \sum_{i\in [\ell],i\neq i'\in [\ell+1],k\in [n]}\frac{1}{M}\lambda_{(R,i,i',k)}(\gamma_{x(z_i)} z_i-\gamma_{k} z_i)&+&\\
\textstyle \sum_{i\in [\ell],x(z_i)\neq k\in [n]}\frac{1}{M}\lambda_{(R,i,i,k)}(\gamma_{x(z_i)} z_i-\gamma_{k} z_i)&+&\\
\textstyle \sum_{i\in [\ell]}\frac{1}{M}\lambda_{(R,i,i,x(z_i))}(\gamma_{x(z_i)} z_i-\gamma_{x(z_i)} z_i)&+&\\
\textstyle \sum_{i\in [\ell],i\neq i'\in [\ell+1],k \in [n]}\frac{1}{M}\lambda_{(L,i+1,i',k)} (\gamma_{x(z_{i})} z_{i+1} - \gamma_{k} z_{i+1})\\
\textstyle \sum_{i\in [\ell],x(z_i)\neq k \in [n]}\frac{1}{M}\lambda_{(L,i+1,i,k)} (\gamma_{x(z_{i})} z_{i+1} - \gamma_{k} z_{i+1})\\
\textstyle \sum_{i\in [\ell]}\frac{1}{M}\lambda_{(L,i+1,i,x(z_i))} (\gamma_{x(z_{i})} z_{i+1} - \gamma_{x(z_i)} z_{i+1})\\
\end{eqnarray*}
Note that by multiplying the above by $M$, the left hand-side is exactly the objective value of the solution $\lambda$, which is known to be strictly positive.
\end{proof}

\section{Randomized Contracts}\label{appx:rand}

In this appendix we discuss randomized contracts. We will distinguish between the following two forms of randomization. First is when given a type report, the principal runs \emph{two independent lotteries}: one is for the action recommendation and the other is for the payment scheme. The agent needs to report truthfully given the lotteries he faces, and follow the realized action for each realized payment scheme. Below we show that this form of randomization cannot provide the principal extra power.

The second form is when the principal runs a \emph{single lottery} over \emph{pairs} of actions and payment schemes. Thus, the random action and random payment are dependent. Here, the agent needs to report truthfully, and follow the action recommendation for each realized pair (see, e.g., Section 12.4 in \citep{ChadeS19}). We end this section by showing that adding this dependency between action and payment \emph{can} increase the principal's revenue. 



\begin{observation}
Let $x:C\to \Delta([n])$ be a (possibly randomized) allocation rule. For every randomized payment rule that implements $x$, there exists a deterministic payment rule that implements $x$ and yields the same utility for the principal. 
\end{observation}

\begin{proof}[Proof Sketch]
Follows from the fact that the agent has quasi-linear utility. Specifically, by linearity of expectation we can replace the randomized payments with the expected payment for each outcome. Then, we can without loss of generality assume that the payment rule is deterministic. 
\end{proof}

%


\begin{proposition}
For every implementable randomized allocation rule, there exists an implementable deterministic allocation rule that gives at least the same utility for the principal. 
\end{proposition}

\begin{proof}
Let $x:C\to\Delta([n])$ be a randomized allocation rule, and let $t: C\to \mathbb{R}^{m+1}_{\geq 0}$ be a contract that implements $x$. Define the following deterministic allocation rule $x'(c)= \arg\max_{\{k\mid x_k(c)>0\}} R_k-T^c_k$.\footnote{Tie-breakings in favor of the high effort action. Any tie-breaking rule will work.} In words, for each type $c$, the allocation rule $x'$ recommends the action which maximizes revenue among all actions that can be recommended by $x$ to the agent, i.e., actions for which $x_k(c)>0$. We will next show that $t$ implements $x'$. Since $t$ implements $x$ and $x'(c)$ has a nonzero probability to be recommended by $x$, it must hold that when facing $t$, $c$'s payoff from reporting truthfully and taking $x'(c)$ is at least as his payoff from reporting non-truthfully or taking any other action. Since we are not changing the payments, we are not affecting the incentives. Thus, $t$ also implements $x'$.
\end{proof}

Next, we introduce a different form of randomization, which potentially generates higher utility for the principal. In this type of contract, $(x,t)$ consists of a distribution over pairs of recommended action and corresponding payments. That is, with slight abuse of notation, $x$ maps $c$ to a distribution over actions $x(c)\in \Delta([n])$, and $t$ maps $c$ to $n$ corresponding payment schemes $(t_1(c),...,t_n(c))$ such that with probability $x_k(c)$ action $k$ is recommended and payments $t_k(c)$ are set forth to type $c$.

\begin{example}\label{ex:rand}
There are two agent types, four actions with required effort levels $\gamma_0=0$, $\gamma_1=1$, $\gamma_2=3$, $\gamma_3=10$, and three outcomes with rewards $r_1=0$, $r_2=20$, $r_3=35$ (that is, $n=m=3$). The distributions over outcomes are $F_0=(1,0,0)$, $F_1=(0,1,0)$, $F_2=(0,0.5,0.5)$, and $F_3=(0,0,1)$. Thus, $R_0=0,R_1=20,R_2=27.5,R_3=35.$ The two types are $c_L = 1$ and $c_H = 3$, and they occur with equal probability.
\end{example}

\begin{proposition}
Consider Example~\ref{ex:rand}. The optimal revenue of a deterministic contract is $19.75$.
\end{proposition}

\begin{proof}
Follows from the min payment LP approach. 
\end{proof}

\begin{proposition}
Consider Example~\ref{ex:rand}. The following randomized contract is IC, with expected revenue of $19.875$. $x(1)=(0,0,0.5,0.5)$, and $t_3(1)=(0,0,14)$, $t_2(1)=(0,1,5)$. $x(1)=(0,1,0,0)$ and $t_1(3)=(0,3,0)$.
\end{proposition}

\begin{proof}
First we show this contract is IC. 
Type $1$'s payoff from reporting truthfully and following the recommended action is as follows $0.5(14-10)+0.5 (3-3)=2.$ He cannot increase his utility by reporting non-truthfully, since the maximum utility he can achieve by facing type $3$'s contract is $2$ (taking action $1$). Further, both $t_3(1)$ and $t_2(1)$ incentivize actions $3$ and $2$ respectively. $t_3(1)$ yields payoff of $4$ from action $3$, while it yields $4$ from action $2$, $-1$ from action $1$, and $0$ from action $0$. $t_2(1)$ yields a utility of $0$ from actions $0,1,2$, and $-5$ from action $3$. Type $3$'s payoff from reporting truthfully and following the recommended action is $3-3=0$. He cannot increase his utility by reporting non-truthfully, since the maximum utility he can achieve from type $1$'s contracts is $0$ (taking action $0$, any other action yields a negative utility). Further, $t_1(3)$ incentivizes action $1$, since it yields payoff of $0$ from action $1$, while yielding $-7.5$ from action $2$, $-30$ from action $3$, and $0$ from action $0$. The principal's utility is $19.875$: $0.25  (35-14)+ 0.25 (27.5-3)+0.5 (20-3).$
\end{proof}

Here we provide and intuitive explanation for the additional power of this form of randomization. Note that in the above example $t_3(1)$ preserves IC for the agent of type $1$: it prevents the agent from reporting non-truthfully by suggesting higher payment, and $t_2(1)$ is for preserving the principal's payoff by suggesting lower payment. Thus, these two contracts ``balance" each other -- each contract takes care of different objective.

\section{Optimal Contract for Uniformly-Distributed Costs}\label{appx:uni}

In this appendix we complete the proof of Theorem~\ref{thm:uniform}. Specifically, we show that under Assumption~\ref{assumption:high-cost-type}, the expected virtual welfare of $x^*$, the welfare maximizing allocation rule (see Definition~\ref{def:welfare-max-alloc}), is strictly less than the expected virtual welfare of the allocation rule defined in~\eqref{eq:new-alloc}.

The proof relies on the following lemmas.

\begin{lemma}\label{lemma:higher-reward} For the randomized allocation rule $x$ defined in~\eqref{eq:new-alloc}, and for every $\epsilon>0$. The expected reward from $x$ is weakly higher than the expected reward of $x^*$, i.e., $\mathbb{E}_{c\sim G}[R_{x(c)}]\geq \mathbb{E}_{c\sim G}[R_{x^*(c)}]$.
\end{lemma}

\begin{lemma}\label{lemma:lower-cost}
There exists $\epsilon>0$ such that for the randomized allocation rule $x$ defined in~\eqref{eq:new-alloc}, the expected cost of $x$ is strictly less than the expected reward of $x^*$, i.e., $\mathbb{E}_{c\sim G}[
\varphi(c) \gamma_{{x}(c)}]<\mathbb{E}_{c\sim G}[\varphi(c) \gamma_{x^*(c)}]$.
\end{lemma}

Combining the above two lemmas we have that for the value of $\epsilon$ specified in Lemma~\ref{lemma:lower-cost}, $\mathbb{E}_{c\sim G}[R_{x(c)}-\varphi(c) \gamma_{{x}(c)}]\geq \mathbb{E}_{c\sim G}[R_{x^*(c)}-\varphi(c) \gamma_{x^*(c)}]$.

\begin{proof}[Proof of Lemma~\ref{lemma:higher-reward}]
Recall that the construction of $x$ uses the set of  breakpoints of the virtual welfare maximizing allocation rule $Z$, and the weights $\lambda_{(L,i+1,{i'},k)},$ $\lambda_{(R,i,{i'},k)} \geq 0$ $\forall i\in [\ell],i'\in[\ell+1],k \in [n]$ which satisfy $\sum_{i'\in [\ell+1],k\in [n]}\lambda_{(L,i+1,i',k)}= \sum_{i'\in [\ell+1],k\in [n]}\lambda_{(R,i,i',k)}=1$ $\forall i\in [\ell]$ and 
\begin{enumerate}
    \item[(1)] $\sum_{i'\in [\ell],k\in [n]}\frac{1}{2}(\lambda_{(R,i',i,k)}F_{k,j}+ \lambda
_{(L,i'+1,i,k)}F_{k,j})\geq 
F_{x^*(z_i),j}$ $\forall  1\le j \le m, i\in [\ell],$
\item[(2)] $\sum_{i\in [\ell]}(\gamma_{x^*(z_i)} z_i+\gamma_{x^*(z_{i})} z_{i+1}) >$ $\sum_{i\in [\ell],i'\in [\ell+1],k\in [n]}(\lambda_{(R,i,i',k)}\gamma_{k} z_i+\lambda_{(L,i+1,i',k)}\gamma_{k} z_{i+1}).$
\end{enumerate}

Using this, we split the integral $\mathbb{E}_{c\sim G}[R_{x}(c)]=\int_{0}^{\bar{c}}{R_{x(c)}g(c)  \dd c}$ by the set of breakpoints and get
\begin{eqnarray*}
\textstyle \mathbb{E}_{c\sim G}[R_{x}(c)]=\sum_{i\in [\ell]}(\int_{z_{i}}^{z_{i}+\epsilon}R_{x(c)} g(c) \dd c +\int_{z_{i}+\epsilon}^{z_{i+1}-\epsilon}R_{x(c)} g(c) \dd c +\int_{z_{i+1}-\epsilon}^{z_{i+1}}R_{x(c)}g(c) \dd c).
\end{eqnarray*}
From the definition of $x$,
\begin{eqnarray*}
\textstyle \mathbb{E}_{c\sim G}[R_{x}(c)]&=& \textstyle \sum_{i\in [\ell]}\int_{z_{i}}^{z_{i}+\epsilon}\sum_{i'\in [\ell],k\in [n]}\lambda_{(R,i,i',k)}R_k g(c) \dd c +\\
\textstyle && \textstyle \sum_{i\in [\ell]}\int_{z_{i}+\epsilon}^{z_{i+1}-\epsilon}R_{x^*(c)} g(c) \dd c +\sum_{i\in [\ell]}\int_{z_{i+1}-\epsilon}^{z_{i+1}}\sum_{i'\in [\ell],k\in [n]}\lambda_{(L,i+1,i',k)}R_k g(c) \dd c.
\end{eqnarray*}
Since $G=U[0,\bar{c}]$, it holds that $\int_{c_1}^{c_2}g(c)\dd c=\frac{c_2-c_1}{\bar{c}}$ $\forall c_1<c_2\in C$. Therefore, 
\begin{eqnarray*}
\textstyle \mathbb{E}_{c\sim G}[R_{x}(c)]= \textstyle  \frac{\epsilon}{\bar{c}}\sum_{i\in [\ell],i'\in [\ell+1],k\in [n]}(\lambda_{(R,i,i',k)}R_k + \lambda_{(L,i+1,i',k)}R_k) + \textstyle \sum_{i\in [\ell]}\int_{z_{i}+\epsilon}^{z_{i+1}-\epsilon}R_{x^*(c)} g(c) \dd c.
\end{eqnarray*}
Using the same arguments, we have the following equality for $x^*$.
\begin{eqnarray*}
\textstyle \mathbb{E}_{c\sim G}[R_{x^*}(c)]= \textstyle  \frac{\epsilon}{\bar{c}}\sum_{i\in [\ell]}2R_{x^*(z_i)} + \textstyle \sum_{i\in [\ell]}\int_{z_{i}+\epsilon}^{z_{i+1}-\epsilon}R_{x^*(c)} g(c) \dd c.
\end{eqnarray*}
It follows that $\mathbb{E}_{c\sim G}[R_{x(c)}]\geq \mathbb{E}_{c\sim G}[R_{x^*(c)}]$ if
\begin{eqnarray*}
\textstyle \sum_{i\in [\ell],i'\in [\ell+1],k\in [n]}(\lambda_{(R,i,i',k)}R_k  + \lambda_{(L,i+1,i',k)}R_k )\geq \sum_{i\in [\ell]}2R_{x^*(z_i)}.
\end{eqnarray*}
By multiplying the inequalities in condition (1) by $r_j$, and summing over $j\in [m]$ (note that $r_0=0$) we have $\sum_{i'\in [\ell],k\in [n]}\frac{1}{2}(\lambda_{(R,i',i,k)}R_k+ \lambda
_{(L,i'+1,i,k)}R_k)\geq 
R_{x^*(z_i)}$ $\forall  1\le j \le m, i\in [\ell].$ This completes the proof.
\end{proof}

\begin{proof}[Proof of Lemma~\ref{lemma:lower-cost}]
Similarly to the proof of Lemma~\ref{lemma:higher-reward}, we split the integral $\mathbb{E}_{c\sim G}[R_{x}(c)]=\int_{0}^{\bar{c}}{R_{x(c)}g(c)  \dd c}$ by the set of breakpoints, we have that  $\mathbb{E}_{c\sim G}[\varphi(c)\gamma_{x(c)}]$, the expected virtual cost of $x$, is 
\begin{eqnarray*}
\textstyle \sum_{i\in [\ell]}(\int_{z_{i}}^{z_{i}+\epsilon}\varphi(c)\gamma_{x(c)} g(c) \dd c +\int_{z_{i}+\epsilon}^{z_{i+1}-\epsilon}\varphi(c)\gamma_{x(c)} g(c) \dd c +\int_{z_{i+1}-\epsilon}^{z_{i+1}}\varphi(c)\gamma_{x(c)}g(c) \dd c).
\end{eqnarray*}
From the definition of $x$, 
\begin{eqnarray*}
\textstyle \mathbb{E}_{c\sim G}[\varphi(c)\gamma_{x(c)}] &=&
\textstyle \sum_{i\in [\ell]}\int_{z_{i}}^{z_{i}+\epsilon}\sum_{i'\in [\ell+1],k\in [n]}\lambda_{(R,i,i',k)}\varphi(c)\gamma_k g(c) \dd c +\\
&& \textstyle \sum_{i\in [\ell]}\int_{z_{i}+\epsilon}^{z_{i+1}-\epsilon}\varphi(c)\gamma_{x^*(c)} g(c) \dd c + \\
&& \textstyle \sum_{i\in [\ell]}\int_{z_{i+1}-\epsilon}^{z_{i+1}}\sum_{i'\in [\ell+1],k\in [n]}\lambda_{(L,i+1,i',k)}\varphi(c)\gamma_k g(c) \dd c.
\end{eqnarray*}
Using integration by parts it can be verified that $\int_{c_1}^{c_2} \varphi(c)g(c) \dd c =   G(c)c\mid^{c_2}_{c_1}$ $\forall c_1,c_2 \in C$. Therefore, 
\begin{eqnarray*}
\textstyle \mathbb{E}_{c\sim G}[\varphi(c)\gamma_{x(c)}] &=&
\textstyle \sum_{i\in [\ell]} \sum_{i'\in [\ell+1],k\in [n]}\lambda_{(R,i,i',k)}\gamma_k G(c)c \mid_{z_{i}}^{z_{i}+\epsilon} + \textstyle \sum_{i\in [\ell]}\gamma_{x^*(c)} G(c)c \mid_{z_{i}+\epsilon}^{z_{i+1}-\epsilon}+ \\
&& \textstyle \sum_{i\in [\ell]}\sum_{i'\in [\ell+1],k\in [n]}\lambda_{(L,i+1,i',k)}\gamma_k  G(c)c \mid_{z_{i+1}-\epsilon}^{z_{i+1}}.
\end{eqnarray*}
Note that $G(c)c\mid^{z_i+\epsilon}_{z_i}=\frac{z_i+\epsilon}{\bar{c}}(z_i+\epsilon)-\frac{z_i}{\bar{c}}z_i=2\frac{z_i \epsilon}{\bar{c}}+\frac{\epsilon^2}{\bar{c}}$, $G(c)c\mid^{z_{i+1}}_{z_{i+1}-\epsilon}=\frac{z_{i+1}}{\bar{c}}z_{i+1}-\frac{z_{i+1}-\epsilon}{\bar{c}}(z_{i+1}-\epsilon)=\frac{2z_{i+1}\epsilon}{\Bar{c}}-\frac{\epsilon^2}{\Bar{c}}.$
Thus, 
\begin{eqnarray*}
\textstyle \mathbb{E}_{c\sim G}[\varphi(c)\gamma_{x(c)}] &=&
\textstyle \frac{\epsilon}{\Bar{c}}\sum_{i\in [\ell]} \sum_{i'\in [\ell+1],k\in [n]}\lambda_{(R,i,i',k)}\gamma_k (2 z_i+\epsilon) + \textstyle \sum_{i\in [\ell]}\gamma_{x^*(c)} G(c)c \mid_{z_{i}+\epsilon}^{z_{i+1}-\epsilon}+ \\
&& \textstyle\frac{\epsilon}{\Bar{c}}\textstyle \sum_{i\in [\ell]}\sum_{i'\in [\ell+1],k\in [n]}\lambda_{(L,i+1,i',k)}\gamma_k (2 z_{i+1}-\epsilon).
\end{eqnarray*}
Using the same arguments, we have the following equality for $\mathbb{E}_{c\sim G}[\varphi(c)\gamma_{x^*(c)}]$.
\begin{eqnarray*}
\textstyle \mathbb{E}_{c\sim G}[\varphi(c)\gamma_{x^*(c)}] &=&
\textstyle \frac{\epsilon}{\Bar{c}}\sum_{i\in [\ell]} \gamma_{x^*(z_i)} (2 z_i+\epsilon) + \textstyle \sum_{i\in [\ell]}\gamma_{x^*(c)} G(c)c \mid_{z_{i}+\epsilon}^{z_{i+1}-\epsilon}+ \\
&&\textstyle \frac{\epsilon}{\Bar{c}}\textstyle \sum_{i\in [\ell]}\gamma_{x^*(z_i)}  (2 z_{i+1}-\epsilon).
\end{eqnarray*}
It follows that $\mathbb{E}_{c\sim G}[\varphi(c)\gamma_{x^*(c)}]<\mathbb{E}_{c\sim G}[\varphi(c)\gamma_{x(c)}]$ if 
\begin{eqnarray*}
\textstyle \sum_{i\in [\ell]} \sum_{i'\in [\ell+1],k\in [n]}\lambda_{(R,i,i',k)}\gamma_k (2 z_i+\epsilon) +
 \textstyle \textstyle \sum_{i\in [\ell]}\sum_{i'\in [\ell+1],k\in [n]}\lambda_{(L,i+1,i',k)}\gamma_k (2 z_{i+1}-\epsilon)&<&\\
\textstyle \sum_{i\in [\ell]} \gamma_{x^*(z_i)} (2 z_i+\epsilon) +
 \textstyle\textstyle \sum_{i\in [\ell]}\gamma_{x^*(z_i)}  (2 z_{i+1}-\epsilon)
\end{eqnarray*}
Rearranging, the above inequality holds if
\begin{eqnarray*}
\textstyle \epsilon \sum_{i\in [\ell]} \sum_{i'\in [\ell+1],k\in [n]}(\lambda_{(R,i,i',k)}\gamma_k  -
\lambda_{(L,i+1,i',k)}\gamma_k)
&<&\\
\textstyle 2\sum_{i\in [\ell]} (\gamma_{x^*(z_i)}  z_i+\gamma_{x^*(z_i)}  z_{i+1})-
\textstyle 2\sum_{i\in [\ell],i'\in [\ell+1],k\in [n]}(\lambda_{(R,i,i',k)}\gamma_k  z_i +\lambda_{(L,i+1,i',k)}\gamma_k  z_{i+1}) &&
\end{eqnarray*}
The right hand-side is strictly positive by (2). Therefore, choosing 
\begin{eqnarray*}
\textstyle \epsilon 
&<& \frac{2\sum_{i\in [\ell]} (\gamma_{x^*(z_i)}  z_i+\gamma_{x^*(z_i)}  z_{i+1})-
\textstyle 2\sum_{i\in [\ell],i'\in [\ell+1],k\in [n]}(\lambda_{(R,i,i',k)}\gamma_k  z_i +\lambda_{(L,i+1,i',k)}\gamma_k  z_{i+1})}{\abs{\sum_{i\in [\ell]} \sum_{i'\in [\ell+1],k\in [n]}(\lambda_{(R,i,i',k)}\gamma_k  -
\lambda_{(L,i+1,i',k)}\gamma_k)}}
\end{eqnarray*}
completes the proof.
\end{proof}

\section{Beyond Uniformly-Distributed Costs}\label{appx:beyond-uni}

In this appendix we present additional observations about other natural classes of distributions, which lead to additional computational results and structural insights, and may prove useful in future work.

\paragraph{\bf Regular distributions}

The first class of distributions is the ``reverse'' version of Myerson's regular distributions definition.

\begin{definition}
A distribution $G$ is \emph{regular} if the corresponding virtual cost $\varphi(\cdot)$ function is nondecreasing. 
\end{definition}

As one would expect, regular distributions give a sufficient condition for the monotonicity of the optimal allocation rule. This implies that when the distribution is regular, the optimal allocation rule is monotone piecewise constant. Combining this with Corollary~\ref{cor:poly-LP-cont} yields the following.

\begin{corollary}
Assuming regular distributions. The problem of determining whether or not the optimal allocation $x^*$ is implementable is solvable in time $O(poly(n,m,|C|^2))$. Moreover, if $x^*$ is implementable, then we can find optimal payments for this allocation rule in time $O(poly(n,m,|C|^2))$.
\end{corollary}

\paragraph{\bf Diminishing marginal returns.}

We give another regularity assumption---which we could not find in the contracts literature. Namely, ``decreasing marginal returns''. 

\begin{definition}
Distributions $F$, rewards $r$ and required efforts $\gamma$ satisfy \emph{diminishing marginal returns} if $\frac{R_i-R_{i-1}}{\gamma_i-\gamma_{i-1}}$, is (weakly) decreasing with $i\in [n]$.  
\end{definition}

In words, the ratio between the gained revenue when shifting from action $i-1$ to action $i$ and the added amount of effort is decreasing with $i\in [n]$. This natural assumption, along with regular distributions, suggests the following nice interpretation of the optimal allocation function.


\begin{theorem}\label{thm:regular-implies-monotone}
Assuming regular distributions and diminishing marginal return, the revenue maximizing rule $x^*$ is monotone piecewise constant, and for every breakpoint $z_i$ for $1<i\in [n]$
$$
\varphi(z_{i})=\frac{R_{n+1-i}-R_{n-i}}{\gamma_{n+1-i}-\gamma_{n-i}}.
$$
\end{theorem}

\begin{proof}
According to Definition~\ref{def:welfare-max-alloc}, it suffices to show that for every $c \in [z_{i},z_{i+1})$, the virtual welfare of action $n-i$, $R_{n-i}-\varphi(c)\cdot \gamma_{n-i}$, is (weakly) higher than that of any other action $k\in [n]$.
There are two analogue cases depending on whether $k<n-i$ or not. Starting with $k<n-i$, we prove that $R_{k}-\varphi(c)\cdot \gamma_{k}\leq R_i-\varphi(c)\cdot \gamma_i$ by downward induction on~$k$. Initially, $k=n-i-1$. Since $c\leq z_{i+1}$ and the virtual costs are non-decreasing; $\varphi(c)\leq \varphi(z_{i+1})$. By the choice of $z_i$, this implies that $\varphi(c)\leq \frac{R_{n-i}-R_{n-i-1}}{\gamma_{n-i}-\gamma_{n-i-1}}$. That is, $R_{n-i-1}-\varphi(c)\cdot \gamma_{n-i-1}\leq R_{n-i}-\varphi(c)\cdot \gamma_{n-i}$. This completes the base case.
For the inductive step, suppose $R_{k+1}-\varphi(c)\cdot \gamma_{k+1}\leq R_{n-i}-\varphi(c)\cdot \gamma_{n-i}$ for $k+1<n-i$. By diminishing marginal returns, $\frac{R_{n-i}-R_{n-i-1}}{\gamma_{n-i}-\gamma_{n-i-1}}\leq \frac{R_{k+1}-R_{k}}{\gamma_{k+1}-\gamma_{k}}$. As shown above, $\varphi(c)\leq \frac{R_{n-i}-R_{n-i-1}}{\gamma_{n-i}-\gamma_{n-i-1}}$. Combining the last two inequalities, we have $R_k-\varphi(c)\cdot \gamma_k \leq R_{k+1}-\varphi(c)\cdot \gamma_{k+1}$. Which, by the induction hypothesis, implies that $R_k-\varphi(c)\cdot \gamma_k \leq R_{n-i}-\varphi(c)\cdot \gamma_{n-i}$. The proof for $k>n-i$ is analogue.
\end{proof}

\begin{remark}
We can show that decreasing marginal returns implies that every action is implementable in the untyped principal-agent problem. Notice that this implication goes away when (single-dimensional) types are added (see Proposition~\ref{prop:non-monotone}).
\end{remark}

\end{document}